\documentclass[12pt, a4paper, oneside]{article}
\def\bSig\mathbf{\Sigma}

\usepackage{float}

\usepackage{amssymb}
\usepackage{amsmath}

\usepackage{multirow}
\usepackage{subcaption}

\usepackage{bbold}

\usepackage{tikz}
\usetikzlibrary{shapes,arrows}
\tikzstyle{block} = [rectangle, draw,
    text width=5em, text centered, rounded corners, minimum height=4em]
\tikzstyle{line} = [draw, -latex']    
\tikzstyle{optional}=[dashed,fill=gray!50]
\tikzstyle{VineNode} = [ellipse, fill = white, draw = black, text = black, align = center, minimum height = 1cm, minimum width = 1cm]
\tikzstyle{DummyNode}  = [draw = none, fill = none, text = black] 
\tikzstyle{TreeLabels} = [draw = none, fill = none, text = black] 
\newcommand{\xshiftNodes}{0.7*\linewidth}
\newcommand{\yshiftLabels}{-.25cm}  
\newcommand{\labelsize}{\footnotesize} 

\usepackage{xcolor, soul}
\sethlcolor{lightgray}  
\usepackage{bbold}
\usepackage{bm}

\DeclareMathOperator*{\argmin}{arg\,min}
\DeclareMathOperator*{\argmax}{arg\,max}

\newtheorem{example}{Example}[section]
\newtheorem{definition}{Definition}[section]
\newtheorem{proposition}{Proposition}[section]
\newtheorem{proof}{Proof}[section]
\usepackage[margin=1in]{geometry}
\usepackage[style=authoryear-comp,maxbibnames=30,maxcitenames=2,
    backend=bibtex,  hyperref=true, firstinits=true,dashed=false]{biblatex} 
\DeclareNameAlias{sortname}{last-first}
 \usepackage[colorlinks = true,
    linkcolor = cyan,
    anchorcolor = cyan,
    citecolor = cyan,
    filecolor = cyan,
    urlcolor = cyan]{hyperref}
 \bibliography{sparsevinereg-2}

\begin{document}
\date{}

\title{High-dimensional sparse vine copula regression \\with application to genomic prediction}

\author{{\"O}zge Sahin\footnote{Corresponding author: ozgesahin-94@hotmail.com}\\ \footnotesize Department of Mathematics, Technical University of Munich, Boltzmanstraße 3, Garching, Germany 
        \and 
        Claudia Czado\\ \footnotesize Department of Mathematics, Technical University of Munich, Boltzmanstraße 3, Garching, Germany
\\ \footnotesize Munich Data Science Institute, Walther-von-Dyck-Stra\ss{}e 10, Garching, Germany}

\maketitle
\vspace{-1cm}
\begin{abstract}
High-dimensional data sets are often available in genome-enabled predictions. Such data sets include nonlinear relationships with complex dependence structures. For such situations, vine copula based (quantile) regression is an important tool. However, the current vine copula based regression approaches do not scale up to high and ultra-high dimensions. To perform high-dimensional sparse vine copula based regression, we propose two methods. First, we show their superiority regarding computational complexity over the existing methods. Second, we define relevant, irrelevant, and redundant explanatory variables for quantile regression. Then we show our method's power in selecting relevant variables and prediction accuracy in high-dimensional sparse data sets via simulation studies. Next, we apply the proposed methods to the high-dimensional real data, aiming at the genomic prediction of maize traits. Some data-processing and feature extraction steps for the real data are further discussed. Finally, we show the advantage of our methods over linear models and quantile regression forests in simulation studies and real data applications.
\newline
    \textit{Keywords}: Genomic prediction; High-dimensional data; Variable selection; Vine copula; Quantile regression
\end{abstract}

\newpage

\section{Introduction}\label{sec:intro}
Genomic prediction (GP) aims at predicting a breeding value using genotypic measurements. Then an unobserved trait is predicted using its genotype information like single-nucleotide polymorphism (SNP). With rapid developments in genomic technologies, researchers have high-dimensional SNP data sets. However, it poses some challenges in prediction modeling, such as a small number of observations and a large number of explanatory variables, skewness in variables, irrelevant and redundant variables, interactions among variables, and nonconstant error variance. To solve the drawbacks regarding the data dimensionality in the GP, statistical or machine learning based approaches have been applied \parencite{li2018genomic}. Recently, quantile regression approaches, which model the conditional distribution of the response, have been utilized to deal with the skewness and outliers in the data \parencite{ perez2020bayesian}. Still, the question has been \textit{how to model conditional quantiles flexibly while handling data dimensionality in GP}. 

It is also important to \textit{identify the SNPs relevant} for predicting breeding values to design future genotype studies. For instance, since the human population has been growing, the stability of food supplies has gained much more importance. Hence, recent plant breeding efforts aim at the genetic improvement of crops. \textcite{holker2019european} provided agronomic measurements and more than 500 thousand SNPs to make European flint maize landraces available for such an aim.

Vine copula based (quantile) regression allows modeling a nonlinear relationship between explanatory variables and response. It considers higher-order explanatory variables and deals with unknown functional error forms. However, the current vine copula based regression methods' computational complexity makes them infeasible to be applied in high-dimensional data sets \parencite{kraus2017d, tepegjozova2022nonparametric}.  We refer to high- and ultra high- dimensional data sets when the number of explanatory variables is between ten and 1000 and larger than 1000, respectively. 

We propose two vine copula based regression methods that perform well in analyzing high-dimensional sparse data sets, where sparsity means that many explanatory variables do not predict the response. Their computational complexity is significantly less than the existing methods. We define relevant, redundant, and irrelevant explanatory variables for quantile regression and assess the methods' prediction power in high-dimensional sparse simulated data sets. Our analyses regarding the inclusion of relevant variables and exclusion of irrelevant variables show our methods' capability to provide sparse models. We apply the methods for genomic prediction of maize traits, proposing data preprocessing and feature extraction steps on the data given by  \textcite{holker2019european}. Such steps can be further applied and improved in future studies. Overall, we can resolve data dimensionality issues in vine copula based regression. To the best of our knowledge, there has not yet been any study performing the genomic prediction using vine copula models, as well as assessing vine copula regression methods' performance in the presence of redundant and irrelevant variables. 

Alternative nonlinear quantile regression models are generalized additive models ($GAM$) \parencite{wood2006generalized}, quantile regression forests ($QRF$) \parencite{meinshausen2006quantile}, and quantile regression neural networks ($QRNN$) \parencite{cannon2011quantile}. $QRNN$ may suffer from quantile crossing \parencite{cannon2018non}, which does not exist in vine copula based approaches by construction. \textcite{kraus2017d} show a better performance of their vine copula based approach than $GAM$. Hence, among nonlinear models, we compare our models with quantile regression forests and show our advantages, especially in the presence of dependent variables. Moreover, despite the quantile crossing problem, we analyze the performance of linear models with variable selection, i.e., linear quantile regression with a LASSO type penalty ($LQRLasso$) \parencite{sherwood2017rqpen}, in nonlinear cases. 

The paper is organized as follows: Section \ref{sec:vineregmethods} introduces vine copulas and new methods; Section \ref{sec:sim} provides simulation studies. We present the real data application in Section \ref{sec:realdata}, discuss our findings and conclude in Section \ref{sec:disc}. 

       \section{High-dimensional sparse vine copula regression}\label{sec:vineregmethods}
    \subsection{D-vine copulas and prediction}\label{sec:concepts}
Copulas are distribution functions, allowing us to separate the univariate margins and dependence structure. Let $\bm{X} =(X_1, \ldots, X_p)^\top \in \mathbb{R}^p$ be a $p$-dimensional random vector with the joint cumulative distribution function (cdf) $F$ and the univariate marginal distributions $F_1, \ldots F_p$. By Sklar's theorem \parencite{sklar1959fonctions}, the copula $C$, corresponding to $F$, is a multivariate cdf with uniform margins such that $F(x_1, \ldots , x_p) = C\big(F_1(x_1), \ldots F_p(x_p)\big)$. When the univariate marginal distributions are continuous, $C$ is unique, which we assume  in the remainder. In addition, the $p$-dimensional joint density $f$ can be written as 
$
f(\bm{x}) = c\big(F_1(x_1), \ldots, F_p(x_p)\big)\cdot f_1(x_1) \cdots f_p(x_p), \bm{x} \in \mathbb{R}^p,
$
where $c$ is the copula density of the random vector $\big(F_1(X_1), \ldots F_p(X_p)\big)^\top \in [0,1]^p$. 

Standard multivariate copulas, such as the multivariate exchangeable Archimedean or Gaussian, often do not accurately model the dependence among the variables. \textcite{aas2009pair} developed a pair-copula construction or vine copula approach using a cascade of bivariate copulas to extend their flexibilities. Such a construction can be represented by an undirected graphical structure involving a set of linked trees, i.e., a regular (R-) vine \parencite{bedford2002vines}. A R-vine for $p$ variables consists of $p-1$ trees, where the edges in the first tree are the nodes of the second tree, the edges of the second tree are the nodes of the third tree and so on. If an edge connects two nodes in the $(t+1)$th tree, their associated edges in the $t$th tree must have a shared node in the $t$th tree for $t=1,\ldots, p-2$. The nodes and edges in the first tree represent $p$ variables and unconditional dependence for $p-1$ pairs of variables, respectively. In the higher trees, the conditional dependence of a pair of variables conditioning on other variables is modeled. To get a vine copula or pair-copula construction, there is a bivariate (pair) copula associated with each edge in the vine. 

One special class of R-vines are D-vines, whose graph structure is a path, i.e., all nodes' degree in the graph is smaller than three. A node in a path represents a variable, and an edge between a pair of nodes corresponds to dependence among the variables of the respective nodes expressed by a pair copula. The node having a degree of one is called a leaf node. Once the order of the nodes in the first tree is determined, the associated D-vine copula decomposition is unique. Moreover, if the pair copulas in the higher tree levels than $t$ are independence copulas, where $t < p$ and $t \geq1$, representing conditional independence, a $t$-truncated vine copula is obtained as shown Appendix in \ref{app1}. 

D-vine copulas allow to express the conditional density of a leaf node in the first tree in a closed form. Here we choose the leaf node in the first tree of a D-vine copula as a response variable. In the remainder, assume that $(y_i, x_{i,1}, \ldots, x_{i,p})$, $i=1, \ldots, n$, are realizations of the random vector $(Y, X_1, \ldots, X_p)$, and $Y$ denotes a response variable with its marginal distributions $F_Y$ and the others correspond to explanatory variables with their marginal distributions $F_1, \ldots, F_p$. For the D-vine copula with the node order $0-1-\ldots-p$ corresponding to the variables $Y-X_1-\ldots-X_p$, $p\geq2$,  as stated in \textcite{kraus2017d}, the conditional quantile function $F^{-1}_{0|1,\ldots,p}$ at quantile $\alpha$ can be expressed in terms of the inverse marginal distribution function $F_{Y}^{-1}$ of the response and the conditional D-vine copula quantile function $C^{-1}_{0|1,\ldots,p}$ at quantile $\alpha$ as 
$F^{-1}_{0|1,\ldots,p}(\alpha|x_1,\ldots, x_p) = F_{Y}^{-1}\big(C^{-1}_{0|1,\ldots,p}(\alpha|F_1(x_1), \ldots, F_p(x_p)\big).$

Since $p$-dimensional D-vine copula's input is the marginally uniform data on $[0,1]^{p}$, the estimation of the D-vine copula follows a two-step approach called the inference for margins \parencite{joe1996estimation}. First, each marginal distribution is estimated. Then the data is converted into the copula data by applying probability integral transformation, e.g., using a univariate non-parametric kernel density estimator, i.e., $\hat{F}_Y$ and $\hat{F}_{X_d}$, $d=1, \ldots, p$. Next, we have the pseudo copula data: $(v_i, u_{i, 1}, \ldots, u_{i,p} )= (\hat{F}_Y(y_i), \hat{F}_{X_1}(x_{i, 1}), \ldots, \hat{F}_{X_p}(x_{i, p}))$, $i=1, \ldots, n$, being realizations of the random vector $(V, U_1, \ldots, U_p)$.   More details about the conditional log-likelihood and estimation of D-vine copulas are given in Appendix \ref{app1}. 
 \subsection{Proposed methods: $vineregRes$ and $vineregParCor$} \label{sec:newmodel}
We propose two methods to perform a D-vine copula regression on high-dimensional sparse data sets: $vineregRes$  and $vineregParCor$.  Appendices \ref{app2} and \ref{app3} give an illustrative example and details of the methods. 
 
The method $vineregRes$ performs the variable selection at a given iteration based on the residuals of the previous iteration, i.e., the pseudo-response. It finds the variable among the candidates, which provides the best bivariate copula conditional log-likelihood conditioned on the variable and conditioning the pseudo-response of the previous iteration. Assume $\tilde{y}^{(s)}_i$ and $\tilde{v}^{(s)}_i$ $i=1, \ldots, n$, denote the pseudo-response and its pseudo copula data in the $s$th iteration, respectively, which are realizations of the random variable $Y^{(s)}$ and $V^{(s)}$, respectively. $V^{(0)}$ and $V^{(s)}$ have the indices 0 and $0^{(s)}$, respectively, and are always a leaf node in the first tree. 

   $vineregRes$
   
    $Step$ 1 (initialization): For given data $(y_i, x_{i,1}, \ldots, x_{i,p})$ and $(v_i, u_{i, 1}, \ldots, u_{i,p} )$, define the initial pseudo-response $\tilde{y}^{(1)}_i=y_i$ with its copula scale $\tilde{v}^{(1)}_i$, $i=1, \ldots, n$, the initial D-vine order $\mathcal{D}^{(1)}=(0)$, the initial chosen variable index set $\mathcal{I}^{(1)}_{var}=\emptyset$, and the initial set of candidate explanatory variables $p^{(1)}_{cand} = \{1, \ldots, p\}$.

    For $s=1, 2, \ldots,$
    
    $Step$ 2 (variable selection): Fit a parametric bivariate copula to data \{$(\tilde{v}^{(s)}_i, u_{i, d})$, $i=1, \ldots, n$\} for $d \in p^{(s)}_{cand}$ and denote the copula, copula density, and its estimated parameters by $\widehat{CR}^{d_{(s)}}$,$\widehat{cr}^{d_{(s)}}$ and $\bm{\hat{\theta}}^{d_{(s)}}$, respectively. Then, find the variable for which the conditional log-likelihood of the copula $\widehat{CR}^{d^{*}_{(s)}}$ is maximized, i.e.,
    
    $
         d^{*}_{(s+1)} = \argmax_{d_{(s)}  \in p^{(s)}_{cand}} \sum_{i=1}^{n} \textrm{ln } \widehat{cr}^{d_{(s)}}_{0^{(s)} | d_{(s)}}(\tilde{v}^{(s)}_i | u_{i, d_{(s)}}; \bm{\hat{\theta}}^{d_{(s)}}).
         $
    
  $Step$ 3 (D-vine extension): Extend the D-vine order by adding the variable with index $d^{*}_{(s+1)}$ to get a D-vine order $\mathcal{D}^{(s+1)}$ = $(\mathcal{D}^{(s)}, d^{*}_{(s+1)})$. Select the parametric pair copula families and estimate the parameters in the extended D-vine structure, where the associated D-vine copula and its estimated parameters are denoted by $\hat{C}^{(s+1)}$ and $\bm{\hat{\theta}}^{(s+1)}$, respectively. 
  
   $Step$ 4 (Chosen variable indices and hyperparameter updates): Extend the chosen variable set, adding the new variable, $\mathcal{I}^{(s+1)}_{var}=\mathcal{I}^{(s)}_{var} \cup d^{*}_{(s+1)}$ and update $p^{(s+1)}_{cand} = p^{(s)}_{cand} \setminus d^{*}_{(s+1)}$.
  
  $Step$ 5 (Pseudo-response update or stop): If a stopping condition
  (Section \ref{sec:stop}) does not hold, estimate the median of the response variable based on the D-vine copula $\hat{C}^{(s+1)}$ and update the pseudo-response, i.e.,
   
   $
    \tilde{y}_i^{(s+1)} =y_i - \hat{F}^{-1}_{Y}(\hat{C}^{-1(s+1)}_{0 | \bm{I}^{(s+1)}_{var}}(0.50 |u_{i, p_1}, \ldots, u_{i, p_{d_{(s+1)}}}; \bm{\hat{\theta}}^{(s+1)} )),  
    $
    \quad$
    \tilde{v}_i^{(s+1)} = \hat{F}_{Y^{(s+1)}}(\tilde{y}_i^{(s+1)} ), 
    $
    
 where $i=1, \ldots, n$ and $\{p_1, \ldots, p_{d_{(s+1)}}\} \subseteq \mathcal{I}^{(s+1)}_{var}$.

Another method to perform an efficient D-vine copula regression for high-dimensional data is to use the partial correlation between the response and a candidate explanatory variable given the chosen variables at each iteration based on their empirical normal scores. Such scores are calculated as detailed in Section 1 of \textcite{joe2014dependence}.

  $vineregParCor$
  
    $Step$ 1 (initialization): As given in $vineregRes$ and the data's normal scores.

    For $s=1, 2, \ldots,$
    
    $Step$ 2 (variable selection): 
    $
         d^{*}_{(s+1)} = \argmax_{d_{(s)}  \in p^{(s)}_{cand}} |{\hat{\rho}_{0, {d_{(s)}  };\mathcal{I}^{(s)}_{var} }}|,
         $
         where $\hat{\rho}_{j,k;S}$ is the estimated partial correlation of variables $j,k$ given those indexed in $S$ based normal scores.
 
  $Step$ 3 (D-vine extension): As given in $vineregRes$.
  
   $Step$ 4 (Chosen variable indices and hyperparameter updates): As given in $vineregRes$.
 
  \subsection{Bivariate copula selection}\label{sec:bicopsel}
Step 3 of $vineregRes$ and $vineregParCor$ selects parametric pair copulas and estimates their parameters associated with the extension of the D-vine structure. Also, Step 2 of $vineregRes$ fits a parametric bivariate copula to the pseudo-response and a candidate explanatory variable. First, we  estimate the parameters that maximize the log-likelihood of a candidate bivariate copula family, which is listed in Appendix \ref{app4}. Later we can select the one with the lowest Akaike (AIC) or the Bayesian information criterion. While extending the D-vine structure and adding new trees at Step 3 of the methods, the fit of parametric pair copulas can be performed sequentially from the lowest to the highest trees \parencite{Brechmann2010}.

  \subsection{Stopping criteria}\label{sec:stop}
To decide if a chosen candidate explanatory variable in a given iteration should be in a model, we will consider the conditional AIC, which penalizes the conditional log-likelihood of the model based on the D-vine copula by the effective degrees of freedom in the model   whose details are provided in Appendix \ref{app4}.   We stop adding variables in D-vine copula regressions when the current iteration's conditional AIC is equal to or larger than the previous iteration's conditional AIC. If the conditional AIC always improves in each iteration, we stop after all explanatory variables are included in the model.
 
  \subsection{Complexity}\label{sec:comp}
Assuming that the data consists of $p$ explanatory variables, the complexity of the existing method $vinereg$ is $\mathcal{O}(p^{3})$ in terms of the total number of bivariate copulas to be selected during the algorithm \parencite{tepegjozova2019d}. Thus, we evaluate the complexity of $vineregRes$ and $vineregParCor$ using the same criterion. We will consider the worst case scenario that the algorithms run until all explanatory variables are included in the model. Further, the total number of estimated parameters is linear in terms of the number of bivariate copulas.  The detailed calculations in Appendix \ref{app4} show that the complexity of $vineregParCor$ and $vineregRes$ in terms of the total number of selected bivariate copulas is $\mathcal{O}(p^{2})$. Hence, our methods significantly reduce the computational complexity of $vinereg$.

 \subsection{Relevant, redundant, and irrelevant variables}\label{sec:red-irr}
Now we define relevant, redundant, and irrelevant variables for predicting the conditional quantile of a response variable $Y$ given the index set of explanatory variables $\mathcal{X}$. We will denote the cdf of the variables with the index set $\mathcal{X}$ by $F_\mathcal{X}$ in the following.

\begin{definition}{(Relevant variables)} The index set of variables $\mathcal{M}  \subseteq \mathcal{X}$ is called relevant for $Y$ if and only if it holds $F_{Y|\mathcal{M}}(y| \bm{x}_\mathcal{M}) \neq F_{Y}(y)$, where $\bm{x}_\mathcal{M}$ includes the variables in $\mathcal{M}$.
\end{definition}
\begin{definition}{(Redundant variables)} The  index set of variables $\mathcal{R}$ is called redundant given the set of variables $\mathcal{M}$ for $Y$ if and only if it holds $F_{Y|\mathcal{M}, \mathcal{R}}(y| \bm{x}_\mathcal{M}, \bm{x}_\mathcal{R}) =F_{Y|\mathcal{M}}(y|\bm{x}_\mathcal{M})$ and $F_{\mathcal{M}, \mathcal{R}}( \bm{x}_\mathcal{M}, \bm{x}_\mathcal{R}) \neq F_{\mathcal{M}}( \bm{x}_\mathcal{M}) \cdot F_{\mathcal{R}}(\bm{x}_\mathcal{R})$, where the vectors $\bm{x}_\mathcal{M}$ and $\bm{x}_\mathcal{R}$ include the variables in the sets $\mathcal{M}$ and $\mathcal{R}$, respectively, $\mathcal{R} \subseteq \mathcal{X}, \mathcal{M} \subseteq \mathcal{X}, \mathcal{R} \cap  \mathcal{M}=\emptyset$.
\end{definition}
A discussion on redundant variables in a D-vine copula is in Appendix \ref{app5}.
\begin{example} Consider the model $(Y,X_1,X_2)^\top \sim \mathcal{N}_3(\bm{0}, \Sigma)$ with (0.5, 0.4, 0.8) vectorizing the upper triangular part of the symmetric covariance matrix $\Sigma$, where it holds $\rho_{Y,X_2;X_1}=0$, i.e., $Y$ is conditionally independent of $X_2$ given $X_1$. Hence, we have $f_{Y| X_1,X_2}(y|x_1,x_2) = \frac{f_{Y,X_2|X_1}(y, x_2|x_1)}{f_{X_2|X_1}(x_2|x_1)} = \frac{f_{Y|X_1}(y|x_1)\cdot f_{X_2|X_1}(x_2|x_1)}{f_{X_2|X_1}(x_2|x_1)}=f_{Y|X_1}(y|x_1)$. \\Since $f_{X_1, X_2} (x_1, x_2) \neq f_{X_1}(x_1) \cdot f_{X_2}(x_2)$, $X_2$ is redundant given $X_1$ for $Y$.
\end{example}

\begin{definition}{(Irrelevant variables)} The set of variables $\mathcal{I}$ is called irrelevant given the set of variables $\mathcal{M}$ for $Y$ if and only if it holds $F_{Y|\mathcal{M}, \mathcal{I}}(y| \bm{x}_\mathcal{M}, \bm{x}_\mathcal{I}) =F_{Y|\mathcal{M}}(y|\bm{x}_\mathcal{M})$, $F_{\mathcal{M}, \mathcal{I}}( \bm{x}_\mathcal{M}, \bm{x}_\mathcal{I}) = F_{\mathcal{M}}( \bm{x}_\mathcal{M}) \cdot F_{\mathcal{I}}(\bm{x}_\mathcal{I})$, and $F_{Y, \mathcal{I}}(y, \bm{x}_\mathcal{I}) = F_Y( y) \cdot F_{\mathcal{I}}(\bm{x}_\mathcal{I})$, where the vectors $\bm{x}_\mathcal{M}$ and $\bm{x}_\mathcal{I}$ include the variables in the sets $\mathcal{M}$ and $\mathcal{I}$, respectively, $\mathcal{I} \subseteq \mathcal{X}, \mathcal{M} \subseteq \mathcal{X}, \mathcal{I} \cap  \mathcal{M}=\emptyset$.
\end{definition}

\begin{example} Consider the model $(Y,X_1,X_2)^\top \sim \mathcal{N}_3(\bm{0}, \Sigma)$ with (0.5, 0, 0) vectorizing the upper triangular part of the symmetric covariance matrix $\Sigma$, where it holds $\rho_{Y,X_2;X_1}=0$; hence, $f_{Y| X_1,X_2}(y|x_1,x_2)=f_{Y|X_1}(y|x_1)$. In addition, it holds $f_{X_1, X_2} (x_1, x_2) = f_{X_1}(x_1) \cdot f_{X_2}(x_2)$ and $f_{Y, X_2} (y, x_2) = f_{Y}(y) \cdot f_{X_2}(x_2)$; thus, $X_2$ is irrelevant given $X_1$ for $Y$.
\end{example}

  \section{Simulation study}\label{sec:sim}
We show the flexibility and effectiveness of the proposed methods on simulated datasets being nonlinear and having different sparsity. We explore the following questions: (Q1) How do $vineregRes$ and $vineregParCor$ work in situations with nonlinear explanatory variable effects on the response's quantiles in the presence of redundant and irrelevant variables for prediction accuracy and computational complexity? (Q2) How well do $vineregRes$ and $vineregParCor$ identify relevant and irrelevant variables for predicting the response's quantiles? (Q3) How do $vineregRes$ and $vineregParCor$ perform compared to  the alternative methods $LQRLasso$ and $QRF$ introduced and discussed in Appendix \ref{app6}? 
 
 \subsection{Data generating process (DGP)} \label{sec:dgp}
\subsection*{DGP1: irrelevant variables}
\begin{equation}
\begin{aligned}
 &Y^{d}_{i} = X_{i,1} \cdot X^2_{i,2} \cdot \sqrt{| X_{i,3}|+0.1} + e^{0.4 \cdot  X_{i,4} \cdot  X_{i,5}}+ \\&(X_{i, 6}, \ldots, X_{i, p_d}) (0, \ldots, 0)^\top+ \epsilon_i \cdot \sigma_i,  \quad i=1, \ldots, n,  \quad d=1,2,3,
 \end{aligned}
\label{eq:sim1}
\end{equation}
where we sample the relevant variables $(X_{i,1}, \ldots, X_{i,5})^\top\sim \mathcal{N}_{5}(\bm{0}, \Sigma)$, $i=1, \ldots, n$ with the $(a, b)$th element of the covariance matrix $\Sigma_{a,b} = 0.75^{|a-b|}$, irrelevant variables \\$(X_{i,6}, \ldots, X_{i,p_d})^\top\sim \mathcal{N}_{p_d-5}(\bm{0}, \mathbb{I}_{p_d-5})$, the random error terms $\epsilon_i \sim \mathcal{N}(0, 1)$ that are independent and identically distributed (iid), independently, and set $\sigma_i\in\{0.5,1\}$, $i=1, \ldots, n$. We simulate data sets with different number of irrelevant variables and set it to $(p_d - 5)$ in each case $d=1,2,3$ as follows: Case 1 with $p_1$ = 10 (50\% of variables are irrelevant),  Case 2 with $p_2$ = 20 (75\% of variables are irrelevant), Case 3 with $p_3$ = 50 (90\% of variables are irrelevant).

\subsection*{DGP2: redundant variables}
\begin{equation}
\begin{aligned}
 &Y^{d}_{i}  = \sqrt{|5 \cdot X_{i,1} - 2 \cdot X_{i,9}+0.5|} + X_{i,8} \cdot (-4 \cdot X_{i,3}+1) +e^{X_{i,6}} + (2\cdot X^{3}_{i,10} +  X^{3}_{i,4})  \\&
 + (X_{i,7}+1)\cdot (\textrm{ln } (| X_{i,2} +  X_{i,5}| + 0.01))+ 
 \\&(X_{i, 11}, \ldots, X_{i, p_d}) (0, \ldots, 0)^\top+ \epsilon_i \cdot \sigma_i,  \quad i=1, \ldots, n,\quad d=1,2,3,4,
 \end{aligned}
\label{eq:sim2}
\end{equation}
where the samples of explanatory variables are independently generated from a multivariate normal distribution with a Toeplitz correlation structure, i.e., $(X_{i,1}, \ldots, X_{i,p_d})^\top\sim \mathcal{N}_{p_d}(\bm{0}, \Sigma)$, $i=1, \ldots, n$, $j=1,2,3$, with the $(a, b)$th element of the covariance matrix $\Sigma_{a,b} = \rho^{|a-b|}$. To represent a challenging but realistic scenario, we set $\rho=0.75$. We sample $\epsilon_i \sim \mathcal{N}(0, 1)$, $i=1, \ldots, n$ (iid) independently from the explanatory variables and set $\sigma_i\in\{0.5,1\}$, $i=1, \ldots, n$. All variables are predicting the response's quantiles; however, given the first ten variables, the others are redundant. We change the number of redundant variables and set it to $(p_d - 10)$ for $d=1,2,3,4$:  Case 1 with $p_1$ = 20 (50\% of variables are redundant given the ten relevant ones), Case 2 with $p_2$ = 40 (75\% of variables are redundant given the ten relevant ones), Case 3 with $p_3$ = 100 (90\% of variables are redundant given the ten relevant ones), Case 4 with $p_4$ = 1000 (99\% of variables are redundant given the ten relevant ones). 

Based on the DGPs in Equations \eqref{eq:sim1} and \eqref{eq:sim2}, we simulate samples with size 450 ($n$=450) with a random split of 300/150 observations for a training/a test set. We replicate our procedure 100 times and average performance measures per sample, which are defined in Appendix \ref{app6}.

\subsection{Performance measures}\label{sec:perfmeasures}
We consider the computation time, the number of chosen variables, \textit{true positive rate (TPR}), and \textit{false discovery rate (FDR)} as methods' performance measures on the training set. To evaluate the performance of a method on the test set, we apply \textit{the pinball loss} ($PL_{\alpha}$) at $\alpha=0.05, 0.50, 0.95$.  TPR is the ratio of the number of chosen relevant variables by a method $M$ to the total number of relevant variables. FDR is the ratio of the number of chosen irrelevant variables to the total number of chosen variables by a method $M$. Higher TPR and smaller FDR are better. $PL_{\alpha}$ measures the accuracy of the quantile predictions $\hat{y}^{\alpha, M}_i$ at the level $\alpha$ by a method $M$ compared to the given response $y_i$,  $i=1,\ldots, n$. The pinball loss is identical to check score, and smaller pinball loss values are better \parencite{steinwart2011estimating}. The formal definitions are given in Appendix \ref{app6}.

\subsection{Results}\label{sec:simresults}
All computations are run on a single node CPU with Intel Xeon Platinum 8380H Processor with around 25 GB RAM, running {\fontfamily{pcr}\selectfont R} version  4.2.2. However, Step 2 of $vineregRes$ can be broken down into parallel fits across candidate variables for a faster computation.

\textit{Variable selection and computational complexity results on the training set}: In Table \ref{tbl:sim-in}, we analyze the TPR and FDR only for the DGP1 setting since all variables are relevant in the DGP2 setting. The computations for $vinereg$ did not complete within three days per replication for the fourth case of the DGP2 setting, making it computationally infeasible. Also, we did not run $LQRLasso$ for that case since it ran around seven hours per replication and had worse performances in the other simulation cases. 

In all cases, $QRF$ chooses all variables in the associated DGP to make predictions. Thus, its TPR is one, the number of selected variables equals the total number of variables in a sample, and its FDR is the proportion of the irrelevant variables in the associated DGP setting. More analyses about $QRF$ are provided in Appendix \ref{app7}.

Excluding $QRF$, in all cases of the DGP1 setting, $vinereg$ has a better TPR performance than the others. However, its FDR is higher than others, adding many irrelevant variables to the model. $vineregRes$ correctly identifies more than 75\% of the relevant variables in all cases of the DGP1 setting. Its FDR is less than 15\% there, making it the best method for FDR. Further, $vineregParCor$'s TPR is higher than 50\% in all DGP1 cases. However, like others, its FDR increases as the number of irrelevant variables increases in the model, reaching more than 50\% in the third DGP1 case. $LQRLasso$ identifies at least 48\% of the relevant variables, but its TPR decreases when the number of irrelevant variables increases. 

While $vineregRes$ selects the lowest number of variables between four and six, $vinereg$ includes almost half of the total number of variables in the data in each case. This highlights the power of $vineregRes$ regarding the exclusion of irrelevant variables in sparse data sets.  $vineregParCor$'s number of chosen variables is between $vinereg$ and $vineregRes$ in all evaluated cases. The same applies to $LQRLasso$, but it selects more variables for estimating median predictions than other quantiles. In an ultra high-dimensional case with 1000 explanatory variables, the number of variables chosen by $vineregParCor$ is, on average, 31.72 with the empirical standard error of 1.23 while it is 7.08 with that of 0.54 for $vineregRes$. 

As the number of variables increases, the average running time for all methods increases. Among vine based methods, $vineregParCor$ provides the fastest computation as expected from the results in Section \ref{sec:comp}. However, $QRF$ provides the fastest computation among all methods considered. $vineregRes$ and $vineregParCor$ run less than 15 minutes in the ultra high-dimensional case. $LQRLasso$'s running time for quantile levels does not differ much.

\begin{table}[h]
\caption{Comparison of the methods' performance on the training set over 100 replications under the cases 1---3 and 1---4 specified in Equations \eqref{eq:sim1} and \eqref{eq:sim2}, respectively. The numbers in parentheses under a method's name column are the corresponding empirical standard errors. (-) shows computational infeasibility. $LQRLasso$ column corresponds to the quantile levels (0.05, 0.50, 0.95). Chosen Vars. corresponds to the total number of chosen variables. Time is in minutes and per replication.}
\centering
\scalebox{0.65}{
\begin{tabular}{cccrrrrr}
\hline
DGP & Measure & Case & $vinereg$ & $vineregRes$ &$vineregParCor$& $QRF$ & $LQRLasso$ (0.05, 0.50, 0.95) \\ \hline
\multirow{12}{*}{1} & \multirow{3}{*}{TPR} & 1 & 0.81 (0.01) & 0.80 (0.02) & 0.67 (0.02) & 1.00 (0.00) & 0.73 (0.02), 0.69 (0.02), 0.63 (0.02) \\ 
 &  & 2 & 0.85 (0.01) & 0.79 (0.03) & 0.59 (0.02) & 1.00 (0.00) & 0.68 (0.02), 0.68 (0.02), 0.55 (0.02) \\ 
 &  & 3 & 0.89 (0.01) & 0.78 (0.02) & 0.56 (0.02) & 1.00 (0.00) & 0.61 (0.02), 0.61 (0.02), 0.48 (0.01) \\ \cline{2-8} 
 & \multirow{3}{*}{FDR} & 1 & 0.28 (0.01) & 0.08 (0.01) & 0.24 (0.02) & 0.50 (0.00) & 0.28 (0.02), 0.32 (0.02), 0.24 (0.02) \\ 
 &  & 2 & 0.55 (0.01) & 0.13 (0.02) & 0.45 (0.02) & 0.75 (0.00) & 0.38 (0.02), 0.49 (0.03), 0.34 (0.03) \\ 
 &  & 3 & 0.80 (0.00) & 0.15 (0.02) & 0.65 (0.02) & 0.90 (0.00) & 0.35 (0.03), 0.58 (0.02), 0.40 (0.03) \\ \cline{2-8} 
 & \multirow{3}{*}{\begin{tabular}[c]{@{}c@{}}Chosen \\ Vars.\end{tabular}} & 1 & 5.83 (0.13) & 4.56 (0.19) & 4.68 (0.16) & 10.00 (0.00) & 5.24 (0.22), 5.48 (0.21), 4.99 (0.26) \\ 
 &  & 2 & 9.76 (0.19) & 5.04 (0.26) & 6.03 (0.23) & 20.00 (0.00) & 6.60 (0.41), 8.48 (0.43), 5.30 (0.36) \\ 
 &  & 3 & 23.24 (0.42) & 5.29 (0.33) & 8.74 (0.30) & 50.00 (0.00) & 5.93 (0.38), 10.20 (0.68), 5.97 (0.47) \\ \cline{2-8} 
 & \multirow{3}{*}{Time} & 1 & 0.18 (0.00) & 0.17 (0.01) & 0.06 (0.00) & 0.01 (0.00) & 0.02 (0.00), 0.02 (0.00), 0.02 (0.00) \\ 
 &  & 2 & 0.97 (0.02) & 0.30 (0.02) & 0.10 (0.01) & 0.01 (0.00) & 0.02 (0.00), 0.43 (0.03), 0.02 (0.00) \\ 
 &  & 3 & 12.06 (0.31) & 0.77 (0.05) & 0.34 (0.03) & 0.03 (0.00) & 0.12 (0.00), 0.14 (0.00), 0.12 (0.00) \\ \hline
\multirow{8}{*}{2} & \multirow{4}{*}{\begin{tabular}[c]{@{}c@{}}Chosen \\ Vars.\end{tabular}}  & 1 & 10.94 (0.23) & 5.41 (0.19) & 6.68 (0.20) & 20.00 (0.00) & 7.05 (0.32), 10.12 (0.36), 8.95 (0.40) \\ 
 &  & 2 & 19.63 (0.54) & 5.17 (0.17) & 8.25 (0.28) & 40.00 (0.00) & 7.88 (0.43), 12.36 (0.51), 9.88 (0.48) \\
 &  & 3 & 62.92 (2.78) & 5.83 (0.22) & 11.66 (0.42) & 100.00 (0.00) & 7.35 (0.36), 15.45 (0.78), 9.44 (0.55) \\
 &  & 4 & - & 7.08 (0.54) & 31.72 (1.23) & 1000.00 (0.00) & - \\ \cline{2-8} 
 & \multirow{4}{*}{Time} & 1 & 1.15 (0.03) & 0.20 (0.01) & 0.16 (0.01) & 0.01 (0.00) & 0.09 (0.00), 0.10 (0.00), 0.08 (0.00) \\ 
 &  & 2 & 7.30 (0.26) & 0.32 (0.01) & 0.29 (0.03) & 0.02 (0.00) & 0.11 (0.00), 0.13 (0.00), 0.11 (0.00) \\ 
 &  & 3 & 159.22 (8.66) & 0.84 (0.03) & 0.72 (0.06) & 0.05 (0.00) & 0.35 (0.00), 0.35 (0.00), 0.34 (0.00) \\ 
 &  & 4 & - & 9.44 (0.69) & 12.18 (1.26) & 0.44 (0.00) & - \\ \hline
\end{tabular}}
\label{tbl:sim-in}
\end{table}

\textit{Prediction accuracy results on the test set}: Table \ref{tbl:sim-out} shows that $vineregRes$ provides the best fit in eight evaluations out of nine (three pinball losses evaluated for three levels) in the DGP1 setting among vine copula based methods. $vinereg$ and $vineregParCor$ have the same accuracy as $vineregRes$ for the first case in the DGP1 setting. However, as the number of irrelevant variables increases, a residual-based variable selection may be better than other vine copula based methods. Moreover, $LQRLasso$ has the highest pinball loss in all cases of the DGP1 setting because of the high nonlinearity in samples. Even though $vineregParCor$'s performance is better than $LQRLasso$, it provides worse fits than the others at the level 95\%. A likely explanation can be that including irrelevant variables in addition to the most relevant ones in a vine copula may negatively impact the prediction accuracy. However, a similar result does not apply to $QRF$. Despite including all irrelevant variables in the model, $QRF$ still performs better than all in seven evaluations out of nine.

Table \ref{tbl:sim-out} shows that  $vineregRes$ provides the lowest pinball loss at three quantile levels in all DGP2 cases, except at the level 95\% in the first case. Since $vineregRes$ gives the most sparse models in the DGP2 setting in Table \ref{tbl:sim-in}, we infer that including many relevant but potentially redundant variables in $vinereg$, $vineregParCor$, and $QRF$ is worsening the prediction accuracy in the DGP2 setting. $LQRLasso$ suffers from nonlinearity. 

Since the relevant, irrelevant, and redundant variables are known in simulation studies, when only the relevant ten variables are used for prediction in the DGP2 setting, $QRF$ has the pinball loss of 0.64, 1.81, and 0.62 at levels $0.05, 0.50,$ and $0.95$,  respectively. Thus,  $vineregRes$ would have better accuracy than $QRF$ in most cases of the DGP2 setting, even if the latter selected the most relevant variables. Thus, $vineregRes$ is more advantageous than $QRF$ in the presence of many dependent variables in our simulations. 

\begin{table}[H]
\centering
\caption{Comparison of the average performance of the methods on the test set for the pinball loss ($PL_\alpha$)  at different quantile levels $\alpha$ over 100 replications under the cases 1---3 and 1---4 specified in Equations \eqref{eq:sim1} and \eqref{eq:sim2}, respectively. The best performance for each quantile level and DGP case is highlighted. The numbers in parentheses under a method's name column are the corresponding empirical standard errors. (-) shows computational infeasibility.}
\scalebox{0.9}{
\begin{tabular}{lllrrrrr}
\hline
DGP & Measure & Case & $vinereg$ & $vineregRes$ & $vineregParCor$& $QRF$ & $LQRLasso$ \\ \hline
\multirow{9}{*}{1} & \multirow{3}{*}{$PL_{0.05}$} & 1 & \hl{0.21 (0.01)} & \hl{0.21 (0.01)} & \hl{0.21 (0.01)} & 0.22 (0.01) & 0.34 (0.01) \\
 &  & 2 & 0.23 (0.01) & \hl{0.22 (0.01)} & \hl{0.22 (0.01)} & \hl{0.22 (0.01)} & 0.34 (0.02) \\
 &  & 3 & 0.24 (0.01) & \hl{0.21 (0.01)} & 0.22 (0.01) & 0.22 (0.01) & 0.32 (0.01) \\ \cline{2-8} 
 & \multirow{3}{*}{$PL_{0.50}$} &  1 & 0.79 (0.04) & 0.79 (0.03) & 0.81 (0.04) & \hl{0.76 (0.04)} & 0.94 (0.02) \\
 &  & 2 & 0.84 (0.04) & 0.79 (0.03) & 0.82 (0.04) & \hl{0.69 (0.02)} & 0.97 (0.04) \\
 &  & 3 & 0.84 (0.02) & 0.76 (0.02) & 0.79 (0.02) & \hl{0.72 (0.02)} & 0.90 (0.02)  \\  \cline{2-8} 
 & \multirow{3}{*}{$PL_{0.95}$} &1 & 0.43 (0.07) & 0.43 (0.06) & 0.46 (0.07) & \hl{0.41 (0.07)} & 0.59 (0.04) \\
 &  & 2 & 0.37 (0.04) & 0.39 (0.04) & 0.44 (0.04) & \hl{0.32 (0.03)} & 0.64 (0.07) \\
 &  & 3 & 0.38 (0.03) & 0.38 (0.03) & 0.41 (0.03) & \hl{0.35 (0.03)} & 0.53 (0.03) \\ \hline
\multirow{12}{*}{2} & \multirow{4}{*}{$PL_{0.05}$} & 1 & \hl{0.53 (0.01)} & \hl{0.53 (0.01)} & 0.54 (0.01) & 0.70 (0.02) & 0.84 (0.02) \\
 &  & 2 & 0.57 (0.01) & \hl{0.54 (0.01)} & 0.56 (0.01) & 0.70 (0.02) & 0.85 (0.02) \\
 &  & 3 & 0.81 (0.04) & \hl{0.54 (0.01)} & 0.58 (0.01) & 0.72 (0.01) & 0.92 (0.03) \\
 &  & 4 & - & \hl{0.55 (0.01)} & 0.89 (0.02) & 0.78 (0.02) & - \\
  \cline{2-8} 
 & \multirow{4}{*}{$PL_{0.50}$} & 1 & 1.87 (0.02) & \hl{1.83 (0.02)} & 1.84 (0.02) & 1.91 (0.02) & 2.20 (0.02) \\
 &  & 2 & 1.99 (0.02) & \hl{1.84 (0.02)} & 1.86 (0.02) & 1.93 (0.03) & 2.26 (0.03) \\
 &  & 3 & 2.59 (0.09) & \hl{1.84 (0.02)} & 1.94 (0.02) & 2.00 (0.03) & 2.29 (0.02) \\
 &  & 4 & - & \hl{1.89 (0.03)} & 2.42 (0.03) & 2.17 (0.03) & - \\ \cline{2-8} 
 & \multirow{4}{*}{$PL_{0.95}$} & 1 & \hl{0.53 (0.01)} & 0.55 (0.01) & 0.55 (0.01) & 0.65 (0.01) & 0.78 (0.02) \\
 &  & 2 & 0.57 (0.01) & \hl{0.56 (0.01)} & \hl{0.56 (0.01)} & 0.67 (0.01) & 0.82 (0.02) \\
 &  & 3 & 0.81 (0.05) & \hl{0.57 (0.01)} & 0.61 (0.02) & 0.68 (0.01) & 0.83 (0.02) \\
 &  & 4 & - & \hl{0.57 (0.02)} & 0.88 (0.03) & 0.75 (0.02) & - \\ \hline
\end{tabular}}
\label{tbl:sim-out}
\end{table}
 
 \section{Application: the genomic prediction of maize traits}\label{sec:realdata}
We describe a real-data application on the doubled-haploid (DH) lines from European flint maize landraces that motivates our methods' usage. \textcite{holker2019european} evaluated 899 DH lines whose data contains genotypic measurements with the SNP array technology and phenotypic measurements of agronomic traits across environments. 

We are interested in the relationship between a DH line’s genotype encoded by its SNPs and its phenotypic outcome described by its traits, i.e., the genomic prediction of maize traits. Specifically, we would like to find relevant SNPs for a trait in a multivariate prediction model using our high-dimensional vine copula regression methods, performing variable selection.

 \subsection{Data description and preprocessing}\label{sec:data-prep}
There are three landraces in the data, and we focus on the Kemater Landmais Gelb (KE) landrace, which has the largest number of observations  (471 out of 899). There are 501,124 explanatory variables, SNPs, which have only zero and two as values, e.g., zero corresponds to the genotype TT, and two denotes the genotype CC.  We predict four responses of agronomic traits separately: early plant height measured by centimetres at the fourth and sixth stages (PH\_V4/V6), female flowering time (FF), and male flowering time (MF) measured by days (Figure \ref{fig:feature-trait}). Plant breeders need to increase the early plant development and avoid decreasing or increasing female and male flowering times during the maize genotype adoption. Thus, the traits' prediction from the genotypic measurements is crucial.

To compare the performance of regression methods, we partition our data randomly into training (67\%) and test (33\%) sets. Then the former and latter contain 314 and 157 observations, respectively. Further, we remove the duplicate explanatory variables, retaining one and the common explanatory variables  with the threshold of 5\%.  For instance, assume an explanatory variable in our training set contains 300 zero values and 14 two values. Then, such a variable does not differ among the observations and might not be expected to have predictive power on a response. Thus, the number of explanatory variables in the training and test sets decreases from 501124 to 44789, i.e., we retain around 9\% of the initial explanatory variables. Hence, the number of observations (DH lines) in the training sets is 314, whereas 147 for their test sets. The number of explanatory variables (SNPs) is 44789 ($p=1,\ldots, 44789$), and there are four univariate responses (traits) ($k=1,\ldots, 4$). 

 \subsection{Feature extraction}\label{sec:feature}
Since our explanatory variables are binary, and there can be associated latent variables with a prediction power on the response, we focus on estimating these latent variables and using them as extracted features in a regression method. Our approach is to group the explanatory variables and estimate their weights in their groups so that such weights are used to estimate the latent variables representing each group. Let $\bm{y}_{k}$ and $\bm{SNP}_{p}$ denote the response vector for trait $k$ and explanatory variable vector $p$, respectively.
\begin{enumerate}
\item Fit a linear regression between a response and an explanatory variable, SNP:

$\bm{y}_{k}= \hat{\beta_{0}}^{p,k} + \hat{\beta_1}^{p,k} \cdot \bm{SNP}_{p}, \quad for  \quad k=1,\ldots,4,  \quad  p = 1, \ldots, 44789$.

\item Perform a two-tailed Wald test for $H_0: \beta_1^{p,k}=0$ versus $H_1: \beta_1^{p,k}\neq0$ and determine the associated $P$-values $P^{p,k}$ for  $k=1,\ldots,4, \quad  p = 1, \ldots, 44789$.

\item Screen the explanatory variables whose $P$-value from the second step are smaller than 0.10 and have the screened set $S_{k}$:

$S_{k}= \{\bm{SNP_{p}}: P^{p,k} < 0.10 | p = 1, \ldots, 44789\} \quad \textrm{for} \quad k=1,\ldots,4$.

\item Order the set of the explanatory variables $S_{k}$ based on their $P$-value non-decreasingly:

$O_{k} = \{\bm{SNP_{w_{1}}}, \ldots, \bm{SNP_{w_{|S_{k}|}}\}}$ with $P^{w_{1},k} \leq ... \leq  P^{w_{|S_{k}|},k}$ for $O_{k} =S_{k}$, $k=1,\ldots,4$.

\item Estimate the latent variables, i.e., create the continuous features $\bm{feature}^{G}_{k, d_k}$ by using a grouping size $G$ of explanatory variables in $O_{k}$ and using their coefficients from (1):

$\bm{feature}^{G}_{k, d_k} = \hat{\beta_1}^{w_{d_k},k}\cdot \bm{SNP_{w_{d_k}}}+\ldots +  \hat{\beta_1}^{w_{d_k+G-1},k}\cdot \bm{SNP_{w_{d_k+G-1}}} \quad \textrm{for}  \quad G \in \{100,200\}$, $\quad n_{k_G}=\left\lceil\frac{ |O_k|}{G}\right\rceil,  \quad  d_k =1, \ldots, n_{k_G}, \quad k=1,\ldots,4$.
\end{enumerate}
Then we have 174 (87) continuous features for FF, 92 (46) continuous features for MF, 198 (99) continuous features for PH\_V4, and 183 (93) continuous features for PH\_V6 by grouping $G=100$ ($G=200$). Figure \ref{fig:feature-trait} shows a scatter plot of a continuous feature and a trait.

   \begin{figure}[H]
      \caption{Scatter plots of an extracted continuous feature, combining 100 SNPs in a feature, and the traits. Feature\_2 corresponds to the combination of the SNPs whose p-value from the OLS of the associated trait is higher than 100 SNPs but lower than others. Similar correspondence applies to other features. Orange curves demonstrate a local polynomial regression fit.}
        \centering
            \includegraphics[width=\textwidth]{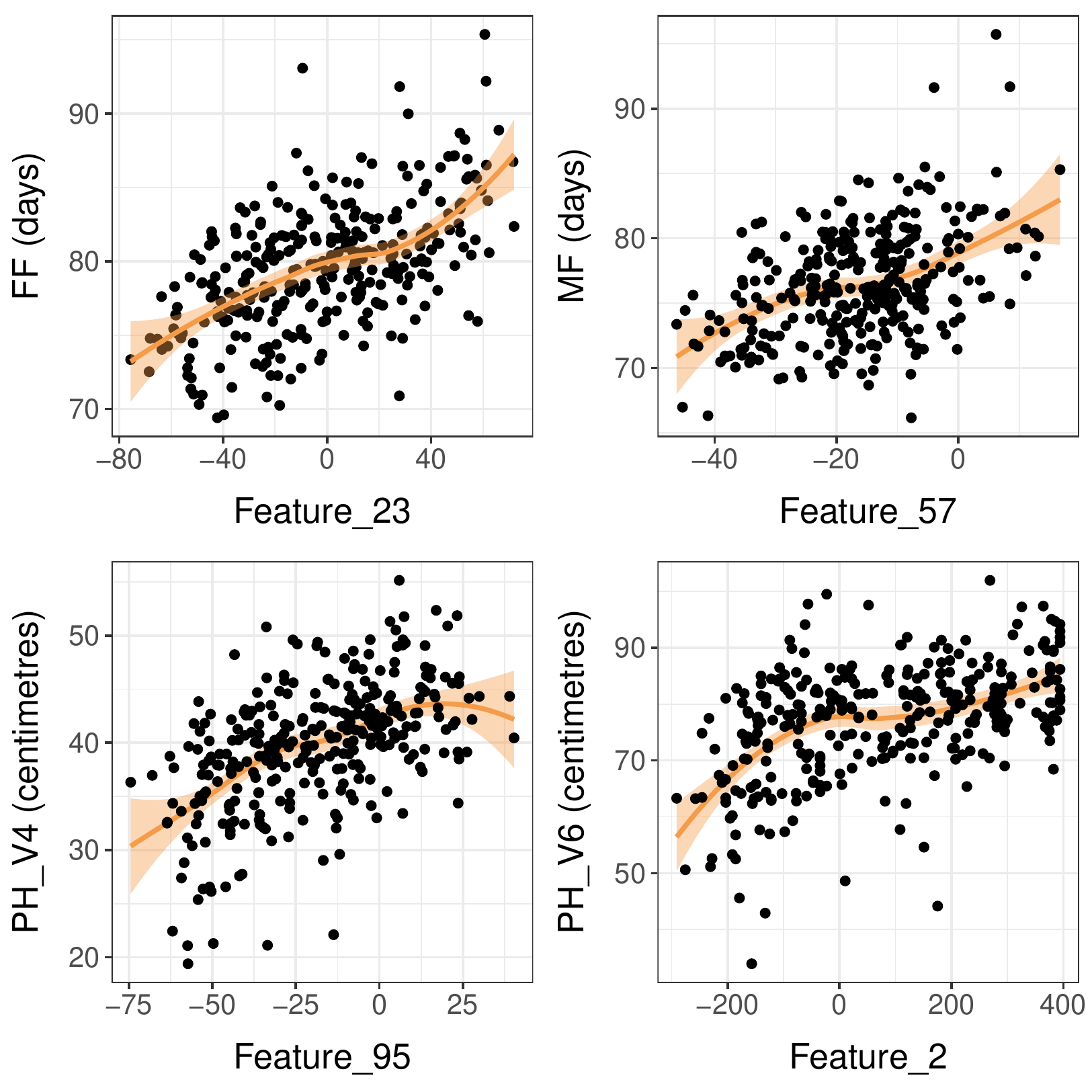}
   \label{fig:feature-trait}
\end{figure}

 \subsection{Prediction}\label{sec:model-fit} 
We have our data $D^{G}_k =(\bm{y}_k, \bm{feature}^{G}_{k,1}, \ldots, \bm{feature}^{G}_{k,n_{k_G}})$ for each response $k=1,\ldots,4$ and $G \in \{100,200\}$.  To identify if a feature is relevant, redundant, or irrelevant, we first conduct a bivariate analysis by fitting a vine copula regression on each feature and trait, i.e., D-vines with two nodes: response and one feature. If a feature is relevant or redundant given the others, our methods add it to the model; otherwise, it is not selected as explained in Section \ref{sec:red-irr}. The bivariate copula family selection between the response and the first feature is conducted as explained in Section \ref{sec:bicopsel}. For instance, we conduct 174 (87) bivariate analyses for FF using a grouping size of  $G=100$ $(200)$. Then all features of four responses are classified as relevant or redundant using  a grouping size of $G=100$ and $G=200$.  Next, we apply our methods to eight different data sets' training sets to find the most relevant features, thereby redundant ones given them. Also, we compare them with $LQRLasso$ and $QRF$ on test sets using the pinball loss defined in Section \ref{sec:perfmeasures} at the levels $0.05, 0.50$, $0.95$. 

Table \ref{tbl:data-res} shows that vine copula based methods perform worse than $LQRLasso$ and $QRF$ for MF. Appendix \ref{app8} shows that dependencies among MF and its selected features by $vineregRes$ are more linear than those among other traits since it fits mostly the Gaussian copula in the first tree for MF. We remark that $LQRLasso$ may perform well if it can avoid crossing quantile curves, but there is no guarantee that the $95\%$ quantile curve exceeds the $90\%$ quantile curve everywhere as shown in Appendix \ref{app9}. Whenever $LQRLasso$ is more accurate than $vineregRes$ for  PH\_V4, it includes more features, giving a trade-off between model sparsity and accuracy. Even though $QRF$ provides the lowest pinball loss at all quantiles for FF for $G=200$, $vineregRes$ has better performance than it for $G=100$, except at the level 95\%. $vineregRes$ is the most sparse and accurate model at all quantile levels considered for PH\_V6 using $G=200$. It chooses two features for $G=200$, identifying more than 95\% of the features as redundant. It has the best accuracy for PH\_V6 for four cases out of six, with three quantile  levels evaluated for two $G$ values.

Given the selected features, the others are redundant for a trait and a grouping of $G$ using $vineregRes$. For instance, given the first and 88th features for PH\_V6 using  a grouping size of $G=200$, the remaining 91 features are redundant using $vineregRes$. Since the features of PH\_V6 are highly dependent but are not needed in a model, in parallel to the simulation study results in Section \ref{sec:simresults}, the reason for our methods' better accuracy than $QRF$ may be many dependent but redundant features for PH\_V6 as detailed in Appendix \ref{app10}.

\begin{table}[H]
  \caption{Comparison of the methods' performance on the test set for the pinball loss ($PL_\alpha$) and on the training set for the number of selected continuous features (No. Ftr.), where (a,b,c) under the $LQRLasso$ column corresponds to the quantile levels (0.05,0.50,0.95). The best performance on the test set for each quantile level $\alpha$, trait, and $G$ is highlighted.}
  \scalebox{0.75}{
\begin{tabular}{llrrrr|rrrr}
\hline
 &  & $vregRes$ & $vregParCor$ & $LQRLasso$ & $QRF$ &$vregRes$ & $vregParCor$ & $LQRLasso$  & $QRF$ \\ \hline
Trait & Measure & \multicolumn{4}{c}{$G=100$} & \multicolumn{4}{c}{$G=200$} \\ \hline
\multirow{4}{*}{FF} &  $PL_{0.05}$ &\hl{0.35}  & 0.49 & 0.40      & 0.38 & 0.39 & 0.39 & 0.39        & \hl{0.37} \\
 & $PL_{0.50}$ & \hl{1.43}  & 1.51 & 1.50      & 1.48 & 1.48 & 1.56 & 1.47        & \hl{1.45} \\
 &  $PL_{0.95}$&  0.47  & 0.47 & \hl{0.41}     & 0.38 & 0.41 & 0.43 & \hl{0.39}        & \hl{0.39}\\
 & No. Ftr. & 11 & 22 & (8,41,4) & 174 & 4 & 14 & (8,29,5) & 87 \\ \hline
\multirow{4}{*}{MF} &  $PL_{0.05}$ &0.35  & 0.36 & 0.34     & \hl{0.33} & 0.35 & 0.36 & \hl{0.32}        & 0.34 \\
 & $PL_{0.50}$ & 1.41  & 1.42 & 1.39     & \hl{1.36} & 1.39 & 1.40  & \hl{ 1.36}        & 1.37 \\
 &  $PL_{0.95}$ & 0.45  & 0.47 & 0.41     & \hl{0.39} & 0.44 & 0.45 & 0.40        &  \hl{0.39} \\
 & No. Ftr. & 12 & 16 & (7,45,8) & 92 & 8 & 13 & (5,15,12) & 46 \\ \hline
\multirow{4}{*}{PH\_V4} &  $PL_{0.05}$ & \hl{0.51}  & \hl{0.51} & 0.55     & 0.55 &  \hl{0.51} &  0.55 & 0.56        & 0.55 \\
 & $PL_{0.50}$ & 1.93  & \hl{1.87} & 1.92     & 1.94 & 1.96 & 1.99 &  \hl{1.92}        & 1.94 \\
 &  $PL_{0.95}$ &0.56  & 0.58 & \hl{0.55}     & 0.60 & 0.57 & 0.57 &  \hl{0.55}       & 0.62 \\
 & No. Ftr.& 6 & 11 & (9,15,8) & 198 & 3 & 11 & (7,17,4) & 99 \\ \hline
\multirow{4}{*}{PH\_V6} &  $PL_{0.05}$ & 1.01  & 1.01 & \hl{0.98}     & 1.00 &  \hl{0.96} & 0.98 & 1.05        & 1.00 \\
 & $PL_{0.50}$ &3.09  & 3.10  & \hl{3.04}     & 3.27 &  \hl{3.06} & 3.47 & 3.14        & 3.31 \\
  &  $PL_{0.95}$& 0.91  & \hl{0.89} & 0.92     & 0.97 & \hl{0.90} & 1.05 & 0.94        & 1.04 \\
 & No. Ftr. & 4 & 12 & (8,49,5) & 183 & 2 & 12 & (6,29,6) & 93\\ \hline
\end{tabular}}
 \normalsize 
\label{tbl:data-res}
\end{table}

Our SNP screening and feature extraction steps are similar to  \textcite{qian2020fast}. Even though they fit a simple linear regression on the first feature, which is based on the linearly and marginally most important SNPs, we conclude in Appendix \ref{app10} that the linearly and marginally most important SNP group might not be considered the most relevant for prediction when allowing nonlinear dependencies as in our methods.

\section{Discussion and conclusion}\label{sec:disc} 
High-dimensional sparse vine copula regression is a significant tool for efficiently allowing nonlinear relationships between explanatory variables and response and selecting relevant variables. In genomic prediction, genotypic measurements like SNPs are often very high-dimensional, which might be reduced by considering some SNP groups and their interactions. Also, many groups may be irrelevant for prediction. Our methods can handle such situations and predict responses at different quantile levels. Their performance might be improved with bivariate copula families having more asymmetries, e.g., more than two parameters.

For our application, consider the following question: Which SNPs impact the low and high quantiles of the trait PH\_V6? $vineregRes$ identifies two SNP groups (features) that consist of 400 SNPs in total. In the first feature, the corresponding SNPs' p-values out of the linear regression with the trait PH\_V6 are in the range [$10^{-12}$, $10^{-7}$], whereas its range is $[0.087, 0.090]$ in the 88th feature. Thus, the marginal impacts of the selected SNPs differ. Given these SNPs, others are redundant to predict the trait PH\_V6.  Thus, plant breeders can assess the selected SNP groups' impact on the trait's various quantile levels and identify the associated SNPs using our methods. Chosen SNP groups can be compared to other genome-wide association studies, helping breeders to decide on future genotype adoption. Hence, comparing the identified SNPs with those in \textcite{mayer2020discovery} is high on the agenda.

Feature extraction is a vital step that may impact our methods' genomic prediction power. For instance, the choice of SNPs' weights estimating their latent variable is open to future research. Also, even though it offers a trade-off between a computational burden and prediction power, one can apply cross-validation for the choice of the SNP's group size $G$. In addition, some SNPs might affect the trait, not marginally only in the presence of certain other SNPs. Alternatively, one may remove the $P$-value screening of the SNPs at the 10\% level described in Section \ref{sec:data-prep} and consider all possible extracted features. Likewise, some SNPs might influence the trait marginally, but not when certain other SNPs are in the model. For such cases, some post-processing steps for feature extraction might be applied. 

Finally, our variable selection steps can be adapted for more flexible vine tree structures.

An {\fontfamily{pcr}\selectfont R}  package, called {\fontfamily{pcr}\selectfont sparsevinereg}, containing the implementation for $vineregRes$ and $vineregParCor$ is given on GitHub: \url{https://github.com/oezgesahin/sparsevinereg}.
\section*{Acknowledgments}
The Munich Data Science Institute funds this research through Seed Funding 2021. Claudia Czado acknowledges the support of the German Research Foundation (DFG grant CZ 86/6-1). We thank Chris-Carolin Sch{\"o}n for funding acquisition, Harry Joe for helpful suggestions, and Manfred Mayer for useful comments on the data. We are grateful to anonymous referees and the associate editor for comments leading to a considerably improved contribution.

\appendix

\section{D-vine copulas}\label{app1}
   \subsection{D-vine example}
   
   \begin{example} Figure \ref{fig:vine} shows the graphical specification for a 4-dimensional D-vine in the form of a set of linked trees. The D-vine with four variables has three trees, and the $j$th tree has $5-j$ nodes and $4-j$ edges, $j = 1, 2, 3$. The first and fourth variables are leaf nodes in the first tree. To model the dependence between a pair of variables, each edge is associated with a pair-copula. Therefore, in the first tree, the pair-copula densities, $c_{1,2}$, $c_{2,3}$, and $c_{3,4}$ model the dependence of the three pairs of variables, (1,2), (2,3), and (3,4), respectively. In the second tree, two conditional dependencies are modeled: between the first and third variables given the second variable, the pair (1,3;2), using associated pair-copula density $c_{1,3;2}$, and between the second and fourth variables given the third variable, the pair (2,4;3), using associated pair-copula density $c_{2,4;3}$. Likewise, the third tree models the conditional dependence between the first and fourth variables given the second and third variables using the associated pair-copula density $c_{1,4;2,3}$.
\begin{figure}[H]
	\centering
	\caption{Example of a 4-dimensional D-vine.}
	\renewcommand{\xshiftNodes}{0.15*\linewidth}
	\renewcommand{\yshiftLabels}{.0cm}  
\begin{tikzpicture}	[every node/.style = VineNode, node distance =1.4cm,scale=0.7, transform shape]
\node[TreeLabels] (T1) {}
node             (v1)         [right of = T1] {1}			
node             (v2)         [right of = v1, xshift = \xshiftNodes] {2}
node             (v3)         [right of = v2, xshift = \xshiftNodes] {3}
node             (v4)         [right of = v3, xshift = \xshiftNodes] {4};
\draw[color=black] (v1) to node[draw=none,  font = \labelsize,fill = none, above, yshift = \yshiftLabels] {1,2} (v2);
\draw[color=black] (v2) to node[draw=none,  font = \labelsize,fill = none, above, yshift = \yshiftLabels] {2,3} (v3);
\draw[color=black] (v3) to node[draw=none,  font = \labelsize,fill = none, above, yshift = \yshiftLabels] {3,4} (v4);
\node[TreeLabels] (T2)      [below of = T1] {}
node             (v12)         [right of = T2, xshift =\xshiftNodes] {1,2}
node             (v23)         [right of = v12, xshift =\xshiftNodes] {2,3}
node             (v34)         [right of = v23, xshift =\xshiftNodes] {3,4};
\draw[color=black] (v12) to node[draw=none,  font = \labelsize,fill = none, above, yshift = \yshiftLabels] {1,3;2} (v23);
\draw[color=black] (v23) to node[draw=none,  font = \labelsize,fill = none, above, yshift = \yshiftLabels] {2,4;3} (v34);
\node[TreeLabels] (T3)      [below of = T2] {}
node             (v132)         [right of = T3, xshift =2*\xshiftNodes] {1,3;2}
node             (v243)         [right of = v132, xshift =\xshiftNodes] {2,4;3};
\draw[color=black] (v132) to node[draw=none,  font = \labelsize,fill = none, above, yshift = \yshiftLabels] {1,4;2,3} (v243);
\end{tikzpicture}
\label{fig:vine}
\end{figure}
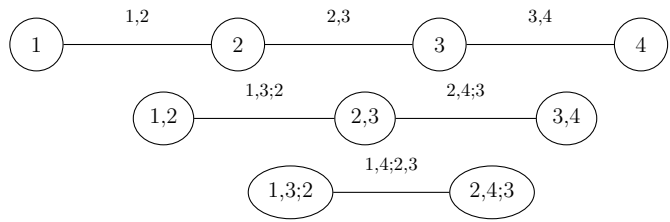

 In addition, if $c_{1,4;2,3}$ is 1 everywhere in the domain, i.e., independence copula, the resulting D-vine is a $2$-truncated vine.
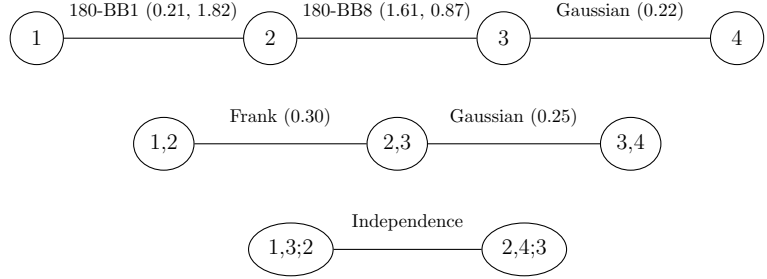
\begin{figure}[H]
	\caption{Example of a 4-dimensional $2$-truncated  D-vine. A letter at an edge with numbers inside the parenthesis refers to its bivariate copula family and rotation with its estimated parameters.}
	\centering

	\renewcommand{\xshiftNodes}{0.15*\linewidth}
	\renewcommand{\yshiftLabels}{.0cm}  
\begin{tikzpicture}	[every node/.style = VineNode, node distance =2cm,scale=0.7, transform shape]
\node[TreeLabels] (T1) {}
node             (v1)         [right of = T1] {1}			
node             (v2)         [right of = v1, xshift = \xshiftNodes] {2}
node             (v3)         [right of = v2, xshift = \xshiftNodes] {3}
node             (v4)         [right of = v3, xshift = \xshiftNodes, fill=none] {4};
\draw[color=black] (v1) to node[draw=none,  font = \labelsize,fill = none, above, yshift = \yshiftLabels] {180-BB1 (0.21, 1.82)} (v2);
\draw[color=black] (v2) to node[draw=none,  font = \labelsize,fill = none, above, yshift = \yshiftLabels] {180-BB8 (1.61, 0.87)} (v3);
\draw[color=black] (v3) to node[draw=none,  font = \labelsize,fill = none, above, yshift = \yshiftLabels] {Gaussian (0.22)} (v4);
\node[TreeLabels] (T2)      [below of = T1] {}
node             (v12)         [right of = T2, xshift =\xshiftNodes] {1,2}
node             (v23)         [right of = v12, xshift =\xshiftNodes] {2,3}
node             (v24)         [right of = v23, xshift =\xshiftNodes, fill=none] {3,4};
\draw[color=black] (v12) to node[draw=none,  font = \labelsize,fill = none, above, yshift = \yshiftLabels] {Frank (0.30)} (v23);
\draw[color=black] (v23) to node[draw=none,  font = \labelsize,fill = none, above, yshift = \yshiftLabels] {Gaussian (0.25)} (v24);
\node[TreeLabels] (T3)      [below of = T2] {}
node             (v132)         [right of = T3, xshift =2*\xshiftNodes, fill=none] {1,3;2}
node             (v243)         [right of = v132, xshift =\xshiftNodes, fill=none] {2,4;3};
\draw[color=black] (v132) to node[draw=none,  font = \labelsize,fill = none, above, yshift = \yshiftLabels] {Independence} (v243);
\end{tikzpicture}
\end{figure}

\end{example}

\subsection{Properties of D-vine copulas}
Assume that $(y_i, x_{i,1}, \ldots, x_{i,p})$, $i=1, \ldots, n$, are realizations of the random vector \\$(Y, X_1, \ldots, X_p)$, and $Y$ denotes a response variable with its marginal distributions $F_Y$ and the others correspond to explanatory variables with their marginal distributions $F_1, \ldots, F_p$. For the D-vine copula with the node order $0-1-\ldots-p$ corresponding to the variables $Y-X_1-\ldots-X_p$, $p\geq2$, the conditional density of the response given the explanatory variables $f_{0|1,\ldots,p}$ can be obtained as:
\begin{equation}
\begin{aligned}
f_{0|1,\ldots,p}(y|x_1,\ldots, x_p)= \big[ \prod_{j=2}^{p} & c_{0, j;1, \ldots, j-1}\big( F_{0 | 1, \ldots, j-1}(y | x_{1}, \ldots, x_{j-1}),  \\ & F_{j | 1, \ldots, j-1}(x_{j} | x_1, \ldots, x_{j-1}\big)\big] \cdot c_{0,1}\big(F_Y(y), F_1(x_1)\big) \cdot f_{Y}(y).
\label{eq:d-vine-cond-dens}
\end{aligned}
\end{equation}
where $F_{0| 1, \ldots, j-1}$ is the conditional distribution function of the random variable $Y|X_{1}=x_{1}, \ldots, X_{j-1} =x_{j-1}$, which can be calculated recursively using h-functions \parencite{joe1996families}. We assume that the pair copulas in the higher tree levels than one do not depend on the conditioning values, e.g., it holds that $c_{0, j;1, \ldots, j-1}(.,.) = c_{0, j;1, \ldots, j-1}(.,.;x_{1}, \ldots, x_{j-1})$. From now on, we make this simplifying assumption. Still, the pair copula density of the high tree levels depends on the conditioning value through its arguments \parencite{stoeber2013simplified}. 

Further, we can calculate the conditional log-likelihood \\$cll_{0|1,\ldots,p}(F_Y(y),F_1(x_1) \ldots, F_p(x_p))$ based on the D-vine copula as follows:
\begin{align*}
cll_{0|1,\ldots,p}(.) = \sum_{j=2}^{p} &\big[\log c_{0, j;1, \ldots, j-1}\big( F_{0 | 1, \ldots, j-1}(y | x_{1}, \ldots, x_{j-1}),  
F_{j | 1, \ldots, j-1}(x_{j} | x_1, \ldots, x_{j-1})\big]
\\ & + \log c_{0,1}\big(F_Y(y), F_1(x_1)) + log f_{Y}(y).
\end{align*}

\section{$vinereg$}\label{app2}

The {\fontfamily{pcr}\selectfont R} package {\fontfamily{pcr}\selectfont vinereg} \parencite{nagler2018vinereg} provides the implementation of a D-vine copula based quantile regression approach of \textcite{kraus2017d}.

Assume the response variable's index is denoted by $0$, and the response variable is a leaf node in the first tree of a D-vine. 
 
    $Step$ 1 (initialization): For the given data $(y_i, x_{i,1}, \ldots, x_{i,p})$ and $(v_i, u_{i, 1}, \ldots, u_{i,p} )$, the initial D-vine order $\mathcal{D}^{(1)}=(0)$, the initial chosen variable index set $\mathcal{I}^{(1)}_{var}=\emptyset$, and the initial set of candidate explanatory variables $p^{(2)}_{cand} = \{1, \ldots, p\}$.

    For $s=1, 2, \ldots,$
    
    $Step$ 2 (variable selection): Extend the D-vine structure order by adding a variable whose index is $d$ to have a D-vine structure order $\mathcal{D}^{(s+1)}$ = $(\mathcal{D}^{(s)}, d)$. Fit a parametric D-vine copula having the structure $\mathcal{D}^{(s+1)}$ and denote the vine copula, its density, and its estimated parameters by $\widehat{CD}^{d_{(s+1)}}$, $\widehat{cd}^{d_{(s+1)}}$, and $\bm{\hat{\theta}}^{d_{(s+1)}}$, respectively. Then find the variable with the index $d^{*}_{(s+1)}$ for which the conditional log-likelihood of the D-vine copula $\widehat{CD}^{d^{*}_{(s+1)}}$ is maximized, i.e.,
    \begin{equation}
         d^{*}_{(s+1)} = \argmax_{d_{(s+1)} \in p^{(s+1)}_{cand}} \sum_{i=1}^{n} \textrm{ln } \widehat{cd}^{d_{(s+1)}}_{0 | d_{(0)}, \ldots, d_{(s+1)}}(v_i | u_{i, d_{(0)}}, \ldots, u_{i, d_{(s+1)}}; \bm{\hat{\theta}}^{d_{(s+1)}}).
         \label{eq:cll-vinereg}
    \end{equation}
  $Step$ 3 (D-vine extension): Extend the D-vine structure order by adding the variable whose index is $d^{*}_{(s+1)}$ to have a D-vine structure order $\mathcal{D}^{(s+1)}$ = $(\mathcal{D}^{(s)}, d^{*}_{(s+1)})$. Select the new parametric pair-copula families and estimate their parameters in the extended D-vine structure. Denote the associated D-vine copula  by $\hat{C}^{(s+1)}$  and its estimated parameters by $\bm{\hat{\theta}}^{(s+1)}$. 
  
   $Step$ 4 (Chosen variable indices and hyperparameter updates): Extend the chosen variable indices $\mathcal{I}^{(s+1)}_{var}=\mathcal{I}^{(s)}_{var} \cup d^{*}_{(s+1)}$ and update $p^{(s+2)}_{cand} = p^{(s+1)}_{cand} \setminus d^{*}_{(s+1)}$.


    \section{$vinereg$, $vineregRes$, and $vineregParCor$}\label{app3}
    \subsection{Illustrative example}
Now we illustrate how $vinereg$, $vineregRes$, and $vineregParCor$ select variables. We simulate the data with 450 observations from the model $(Y,X_1,X_2,X_3, X_4)^\top \sim \mathcal{N}_5(\bm{0},\Sigma)$ with $(0.40, 0.70, 0.32, 0, 0, 0, 0, 0, 0, 0)$ vectorizing the upper triangular part of symmetric matrix $\Sigma$. Thus, there are a response variable and four explanatory variables in the data, and the third and fourth variables are assumed to be irrelevant for predicting the response. 

First, we convert the observations on the copula scale $(v_i, u_{i, 1}, u_{i,2}, u_{i,3}, u_{i,4})$, $i=1, \ldots, 450$, using the non-parametric kernel density estimator. Then we define the initial pseudo-response on the copula scale $\tilde{v}^{(1)}_i=v_i$, $i=1, \ldots, 450$, the initial D-vine order $\mathcal{D}^{(1)}=(0)$, the initial chosen variable index set $\mathcal{I}^{(1)}_{var}=\emptyset$, the initial set of candidate explanatory variables $p^{(1)}_{cand} = \{1,2,3,4\}$, and the given data's normal scores $(z_{i,0}, z_{i, 1}, z_{i,2}, z_{i,3}, z_{i,4})^{\top}$, $i=1, \ldots, 450$. 

For an illustration purpose, we will show the methods' variable selection step in the third iteration, where it holds that the D-vine order $\mathcal{D}^{(3)}=(0, 2, 1)$ for the D-vine copula $\hat{C}^{(3)}$ with the parameters $\bm{\hat{\theta}}^{(3)}$, the chosen variable index set $\mathcal{I}^{(3)}_{var}=\{2, 1\}$, and the set of candidate explanatory variables $p^{(3)}_{cand} = \{3, 4\}$. The conditional AIC of the current D-vine copula $\hat{C}^{(3)}$ shown in Figure \ref{fig:vine-step3} is 975.30. The next iteration for all methods is to decide which and if the third or fourth variable should be added to the model. 

\begin{figure}[H]
	\centering
	\caption{The D-vine copula $\hat{C}^{(3)}$ with the D-vine order $\mathcal{D}^{(3)}=(0, 2, 1)$. A letter at an edge with numbers inside the parenthesis refers to its bivariate copula family and rotation with its estimated parameters.}
	\renewcommand{\xshiftNodes}{0.15*\linewidth}
	\renewcommand{\yshiftLabels}{.0cm}  
\begin{tikzpicture}	[every node/.style = VineNode, node distance =1.4cm,scale=0.7, transform shape]
\node[TreeLabels] (T1) {}
node             (v1)         [right of = T1] {0}			
node             (v2)         [right of = v1, xshift = \xshiftNodes] {2}
node             (v3)         [right of = v2, xshift = \xshiftNodes] {1};
\draw[color=black] (v1) to node[draw=none,  font = \labelsize,fill = none, above, yshift = \yshiftLabels] {180-BB1 (0.21, 1.82)} (v2);
\draw[color=black] (v2) to node[draw=none,  font = \labelsize,fill = none, above, yshift = \yshiftLabels] {180-BB8 (1.61, 0.87)} (v3);
\node[TreeLabels] (T2)      [below of = T1] {}
node             (v12)         [right of = T2, xshift =\xshiftNodes] {0,2}
node             (v23)         [right of = v12, xshift =\xshiftNodes] {1,2};
\draw[color=black] (v12) to node[draw=none,  font = \labelsize,fill = none, above, yshift = \yshiftLabels] {Gaussian (0.22)} (v23);
\end{tikzpicture}
\label{fig:vine-step3}
\end{figure}
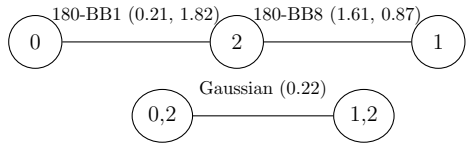

To make the variable selection, $vinereg$ extends the D-vine order as $\mathcal{D}^{(4)}=(0, 2, 1, 3)$ for the third variable and $\mathcal{D}^{(4)}=(0, 2, 1, 4)$ for the fourth  variable. Accordingly, three new pair-copula families are selected for each vine  order, thereby estimating their parameters as shown in Figure \ref{fig:vine-step4}. Thus, $vinereg$ selects six bivariate copulas in this iteration (one for the first/second/third tree level for each D-vine order) to choose one variable. 

\begin{figure}[H]
	\caption{$\hat{CD}^{3_{(4)}}$ with the conditional AIC of 975.30 (left) and $\hat{CD}^{4_{(4)}}$ with the conditional AIC of 975.30 (right) selected and estimated by $vinereg$. A letter at an edge with numbers inside the parenthesis refers to its bivariate copula family and rotation with its estimated parameters. The new selected pair-copula families and nodes compared to the previous step's D-vine copula $\hat{C}^{(3)}$  are shown in orange.}
\begin{subfigure}[b]{0.45\linewidth}
	\centering

	\renewcommand{\xshiftNodes}{0.15*\linewidth}
	\renewcommand{\yshiftLabels}{.0cm}  
\begin{tikzpicture}	[every node/.style = VineNode, node distance =2cm,scale=0.7, transform shape]
\node[TreeLabels] (T1) {}
node             (v1)         [right of = T1] {0}			
node             (v2)         [right of = v1, xshift = \xshiftNodes] {2}
node             (v3)         [right of = v2, xshift = \xshiftNodes] {1}
node             (v4)         [right of = v3, xshift = \xshiftNodes, fill=orange] {3};
\draw[color=black] (v1) to node[draw=none,  font = \labelsize,fill = none, above, yshift = \yshiftLabels] {180-BB1 (0.21, 1.82)} (v2);
\draw[color=black] (v2) to node[draw=none,  font = \labelsize,fill = none, above, yshift = \yshiftLabels] {180-BB8 (1.61, 0.87)} (v3);
\draw[color=orange] (v3) to node[draw=none,  font = \labelsize,fill = none, above, yshift = \yshiftLabels] {Independence} (v4);
\node[TreeLabels] (T2)      [below of = T1] {}
node             (v12)         [right of = T2, xshift =\xshiftNodes] {0,2}
node             (v23)         [right of = v12, xshift =\xshiftNodes] {1,2}
node             (v24)         [right of = v23, xshift =\xshiftNodes, fill=orange] {2,3};
\draw[color=black] (v12) to node[draw=none,  font = \labelsize,fill = none, above, yshift = \yshiftLabels] {Gaussian (0.22)} (v23);
\draw[color=orange] (v23) to node[draw=none,  font = \labelsize,fill = none, above, yshift = \yshiftLabels] {Independence} (v24);
\node[TreeLabels] (T3)      [below of = T2] {}
node             (v132)         [right of = T3, xshift =2*\xshiftNodes, fill=orange] {0,1;2}
node             (v243)         [right of = v132, xshift =\xshiftNodes, fill=orange] {1,3;2};
\draw[color=orange] (v132) to node[draw=none,  font = \labelsize,fill = none, above, yshift = \yshiftLabels] {Independence} (v243);
\end{tikzpicture}
\end{subfigure}
\begin{subfigure}[b]{0.45\linewidth}
	\centering

	\renewcommand{\xshiftNodes}{0.15*\linewidth}
	\renewcommand{\yshiftLabels}{.0cm}  
\begin{tikzpicture}	[every node/.style = VineNode, node distance =2cm,scale=0.7, transform shape]
\node[TreeLabels] (T1) {}
node             (v1)         [right of = T1] {0}			
node             (v2)         [right of = v1, xshift = \xshiftNodes] {2}
node             (v3)         [right of = v2, xshift = \xshiftNodes] {1}
node             (v4)         [right of = v3, xshift = \xshiftNodes, fill=orange] {4};
\draw[color=black] (v1) to node[draw=none,  font = \labelsize,fill = none, above, yshift = \yshiftLabels] {180-BB1 (0.21, 1.82)} (v2);
\draw[color=black] (v2) to node[draw=none,  font = \labelsize,fill = none, above, yshift = \yshiftLabels] {180-BB8 (1.61, 0.87)} (v3);
\draw[color=orange] (v3) to node[draw=none,  font = \labelsize,fill = none, above, yshift = \yshiftLabels] {Independence} (v4);
\node[TreeLabels] (T2)      [below of = T1] {}
node             (v12)         [right of = T2, xshift =\xshiftNodes] {0,2}
node             (v23)         [right of = v12, xshift =\xshiftNodes] {1,2}
node             (v24)         [right of = v23, xshift =\xshiftNodes, fill=orange] {2,4};
\draw[color=black] (v12) to node[draw=none,  font = \labelsize,fill = none, above, yshift = \yshiftLabels] {Gaussian (0.22)} (v23);
\draw[color=orange] (v23) to node[draw=none,  font = \labelsize,fill = none, above, yshift = \yshiftLabels] {Independence} (v24);
\node[TreeLabels] (T3)      [below of = T2] {}
node             (v132)         [right of = T3, xshift =2*\xshiftNodes, fill=orange] {0,1;2}
node             (v243)         [right of = v132, xshift =\xshiftNodes, fill=orange] {1,4;2};
\draw[color=orange] (v132) to node[draw=none,  font = \labelsize,fill = none, above, yshift = \yshiftLabels] {Independence} (v243);
\end{tikzpicture}
\end{subfigure}
\label{fig:vine-step4}
\end{figure}

To perform the variable selection, first, $vineregRes$ estimates the median of the response based on the D-vine copula $\hat{C}^{(3)}$. The analysis of other quantiles than the median using the DGP1 setting in Section 3.1 of the manuscript is given in Section \ref{sec:quantiles}. Then $vineregRes$  calculates the pseudo-response and estimates the pseudo-response's marginal distribution nonparametrically. Next, the associated pseudo-response on the copula scale is obtained by the PIT, i.e.,

    $
    \tilde{y}_i^{(3)} =y_i - \hat{F}^{-1}_{Y}(\hat{C}^{-1(3)}_{0 | 2, 1}(0.50 |u_{i, 2}, u_{i, 1}; \bm{\hat{\theta}}^{(3)} )), \quad 
    \tilde{v}_i^{(3)} = \hat{F}_{Y^{(3)}}(\tilde{y}_i^{(3)} ), \quad i=1, \ldots, 450.
    $
    
Then, it fits a parametric bivariate copula to data \{$(\tilde{v}^{(3)}_i, u_{i, 3})$, $i=1, \ldots, n$\} for the third variable and \{$(\tilde{v}^{(3)}_i, u_{i, 4})$, $i=1, \ldots, n$\} for the fourth variable. It holds that the fitted copula for the third variable $\widehat{CR}^{3_{(3)}}$ is a Frank copula with the estimated parameter $\hat{\theta}^{3_{(3)}}=-0.41$. The conditional log-likelihood of $\widehat{CR}^{3_{(3)}}$ conditioned on the third variable is -474.27. Moreover, the fitted copula for the fourth variable $\widehat{CR}^{4_{(3)}}$ is the independence copula, resulting in the conditional log-likelihood of -475.30. They are shown in Figure \ref{fig:vineRes-step4}. 

\begin{figure}[H]
	\caption{$\hat{CR}^{3_{(3)}}$ with the conditional log-likelihood of -474.20 (left) and $\hat{CR}^{4_{(3)}}$ with the conditional log-likelihood of -474.27 (right) selected and estimated by $vineregRes$. A letter at an edge with numbers inside the parenthesis refers to its bivariate copula family with its estimated parameters. The new selected pair-copula families and nodes are shown in orange.}
\begin{subfigure}[b]{0.45\linewidth}
	\centering
	\renewcommand{\xshiftNodes}{0.15*\linewidth}
	\renewcommand{\yshiftLabels}{.0cm}  
\begin{tikzpicture}	[every node/.style = VineNode, node distance =2cm,scale=0.7, transform shape]
\node[TreeLabels] (T1) {}
node             (v1)         [right of = T1] {$\tilde{\bm{V}}^{(3)}$}			
node             (v2)         [right of = v1, xshift = \xshiftNodes, fill=orange] {3};
\draw[color=orange] (v1) to node[draw=none,  font = \labelsize,fill = none, above, yshift = \yshiftLabels] {Frank (-0.41)} (v2);
\end{tikzpicture}
\end{subfigure}
\qquad
\begin{subfigure}[b]{0.45\linewidth}
	\centering
	\renewcommand{\xshiftNodes}{0.15*\linewidth}
	\renewcommand{\yshiftLabels}{.0cm}  
\begin{tikzpicture}	[every node/.style = VineNode, node distance =2cm,scale=0.7, transform shape]
\node[TreeLabels] (T1) {}
node             (v1)         [right of = T1] {$\tilde{\bm{V}}^{(3)}$}			
node             (v2)         [right of = v1, xshift = \xshiftNodes, fill=orange] {4};
\draw[color=orange] (v1) to node[draw=none,  font = \labelsize,fill = none, above, yshift = \yshiftLabels] {Independence} (v2);
\end{tikzpicture}
\end{subfigure}
\label{fig:vineRes-step4}
\end{figure}
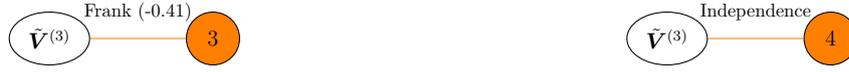

Since the conditional log-likelihood of $\widehat{CR}^{3_{(3)}}$ is higher than that of $\widehat{CR}^{4_{(3)}}$, $vineregRes$ identifies the third variable as the candidate explanatory variable to be added to the D-vine copula. Thus, it extends the D-vine order as $\mathcal{D}^{(4)}=(0, 2, 1, 3)$, selecting the new pair-copula families and estimating new parameters as shown in Figure \ref{fig:vineRes-step4-2}. The resulting D-vine copula $\hat{C}^{(4)}$ has the conditional AIC of  975.30. Since $\hat{C}^{(4)}$ does not have a better conditional AIC than $\hat{C}^{(3)}$, $vineregRes$ stops the iterations and returns the final D-vine copula as $\hat{C}^{(3)}$. Thus, $vineregRes$ chooses five bivariate copulas here.

\begin{figure}[H]
	\caption{The D-vine copula $\hat{C}^{(4)}$ with the conditional AIC of 975.30 selected and estimated by $vineregRes$.  A letter at an edge with numbers inside the parenthesis refers to its bivariate copula family and rotation with its estimated parameters. The new selected pair-copula families and nodes compared to the previous step's D-vine copula $\hat{C}^{(3)}$  are shown in orange.}
	\centering

	\renewcommand{\xshiftNodes}{0.15*\linewidth}
	\renewcommand{\yshiftLabels}{.0cm}  
\begin{tikzpicture}	[every node/.style = VineNode, node distance =2cm,scale=0.7, transform shape]
\node[TreeLabels] (T1) {}
node             (v1)         [right of = T1] {0}			
node             (v2)         [right of = v1, xshift = \xshiftNodes] {2}
node             (v3)         [right of = v2, xshift = \xshiftNodes] {1}
node             (v4)         [right of = v3, xshift = \xshiftNodes, fill=orange] {3};
\draw[color=black] (v1) to node[draw=none,  font = \labelsize,fill = none, above, yshift = \yshiftLabels] {180-BB1 (0.21, 1.82)} (v2);
\draw[color=black] (v2) to node[draw=none,  font = \labelsize,fill = none, above, yshift = \yshiftLabels] {180-BB8 (1.61, 0.87)} (v3);
\draw[color=orange] (v3) to node[draw=none,  font = \labelsize,fill = none, above, yshift = \yshiftLabels] {Independence} (v4);
\node[TreeLabels] (T2)      [below of = T1] {}
node             (v12)         [right of = T2, xshift =\xshiftNodes] {0,2}
node             (v23)         [right of = v12, xshift =\xshiftNodes] {1,2}
node             (v24)         [right of = v23, xshift =\xshiftNodes, fill=orange] {2,3};
\draw[color=black] (v12) to node[draw=none,  font = \labelsize,fill = none, above, yshift = \yshiftLabels] {Gaussian (0.22)} (v23);
\draw[color=orange] (v23) to node[draw=none,  font = \labelsize,fill = none, above, yshift = \yshiftLabels] {Independence} (v24);
\node[TreeLabels] (T3)      [below of = T2] {}
node             (v132)         [right of = T3, xshift =2*\xshiftNodes, fill=orange] {0,1;2}
node             (v243)         [right of = v132, xshift =\xshiftNodes, fill=orange] {1,3;2};
\draw[color=orange] (v132) to node[draw=none,  font = \labelsize,fill = none, above, yshift = \yshiftLabels] {Independence} (v243);
\end{tikzpicture}
\label{fig:vineRes-step4-2}
\end{figure}
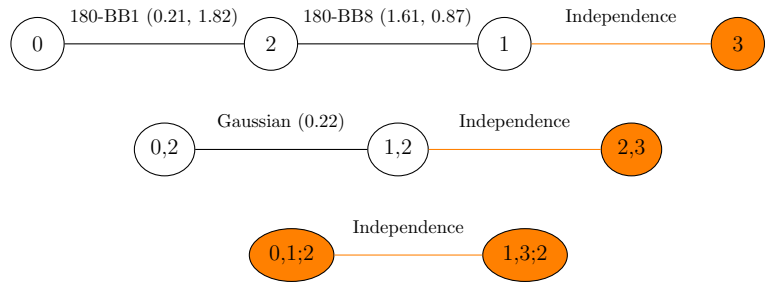

$vineregParCor$ chooses the next variable based on its partial correlation with the response given the chosen explanatory variables using their normal scores. It estimates $\hat{\rho}_{0,3;2,1}=0.05$ for the third variable and $\hat{\rho}_{0,4;2,1}=0.03$ for the fourth variable. Since the third variable has a higher (absolute) estimated partial correlation than the other, $vineregParCor$ extends the D-vine order as $\mathcal{D}^{(4)}=(0, 2, 1, 3)$. Next, it selects the new pair-copula families and performs new parameter estimation associated with $\mathcal{D}^{(4)}$ as shown in Figure \ref{fig:vineParCor-step4}. Since the D-vine copula with $\mathcal{D}^{(4)}=(0, 2, 1, 3)$ does not have a better conditional AIC than $\mathcal{D}^{(3)}=(0, 2, 1)$, $vineregParCor$ gives the final model fit  $\hat{C}^{(3)}$. Thus, $vineregParCor$ selects three bivariate copulas in this iteration.

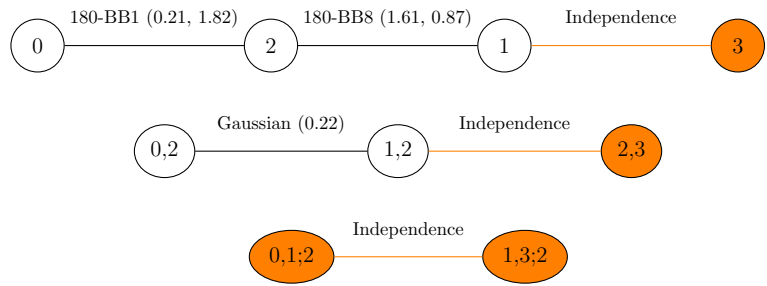
\begin{figure}[H]
	\caption{The D-vine copula $\hat{C}^{(4)}$ with the conditional AIC of 975.30 selected and estimated by $vineregParCor$.  A letter at an edge with numbers inside the parenthesis refers to its bivariate copula family and rotation with its estimated parameters. The new selected pair-copula families and nodes compared to the previous step's D-vine copula $\hat{C}^{(3)}$  are shown in orange.}
	\centering

	\renewcommand{\xshiftNodes}{0.15*\linewidth}
	\renewcommand{\yshiftLabels}{.0cm}  
\begin{tikzpicture}	[every node/.style = VineNode, node distance =2cm,scale=0.7, transform shape]
\node[TreeLabels] (T1) {}
node             (v1)         [right of = T1] {0}			
node             (v2)         [right of = v1, xshift = \xshiftNodes] {2}
node             (v3)         [right of = v2, xshift = \xshiftNodes] {1}
node             (v4)         [right of = v3, xshift = \xshiftNodes, fill=orange] {3};
\draw[color=black] (v1) to node[draw=none,  font = \labelsize,fill = none, above, yshift = \yshiftLabels] {180-BB1 (0.21, 1.82)} (v2);
\draw[color=black] (v2) to node[draw=none,  font = \labelsize,fill = none, above, yshift = \yshiftLabels] {180-BB8 (1.61, 0.87)} (v3);
\draw[color=orange] (v3) to node[draw=none,  font = \labelsize,fill = none, above, yshift = \yshiftLabels] {Independence} (v4);
\node[TreeLabels] (T2)      [below of = T1] {}
node             (v12)         [right of = T2, xshift =\xshiftNodes] {0,2}
node             (v23)         [right of = v12, xshift =\xshiftNodes] {1,2}
node             (v24)         [right of = v23, xshift =\xshiftNodes, fill=orange] {2,3};
\draw[color=black] (v12) to node[draw=none,  font = \labelsize,fill = none, above, yshift = \yshiftLabels] {Gaussian (0.22)} (v23);
\draw[color=orange] (v23) to node[draw=none,  font = \labelsize,fill = none, above, yshift = \yshiftLabels] {Independence} (v24);
\node[TreeLabels] (T3)      [below of = T2] {}
node             (v132)         [right of = T3, xshift =2*\xshiftNodes, fill=orange] {0,1;2}
node             (v243)         [right of = v132, xshift =\xshiftNodes, fill=orange] {1,3;2};
\draw[color=orange] (v132) to node[draw=none,  font = \labelsize,fill = none, above, yshift = \yshiftLabels] {Independence} (v243);
\end{tikzpicture}
\label{fig:vineParCor-step4}
\end{figure}

\subsection{Analysis of Step 5 using different quantiles in $vineregRes$}\label{sec:quantiles}
We run $vineregRes$ using quantile levels 0.05, 0.50, and 0.95 to estimate residuals for the DGP1 setting. We call the models $vineregRes$-0.05,  $vineregRes$-0.50, and  $vineregRes$-0.95, respectively.  We gave the results in Table \ref{tab:alpha}. We observed that using the level 0.05 gave better TPR than the level 0.95. On the other hand, the level 0.95 has a lower FDR than the level 0.05. However, as implemented in the manuscript, the TPR and FDR of the level 0.50 are usually between those of the others. Thus, it averages the others' power in selecting variables. Moreover, the lowest pinball loss was achieved using the level 0.50 in all cases at all quantiles, i.e., using the model $vineregRes$-0.50.

\begin{table}[H]
\centering
\caption{Comparison of the average performance of the methods $vineregRes$ using quantile levels 0.05 ($vineregRes$-0.05), 0.50 ($vineregRes$-0.50), and 0.95 ($vineregRes$-0.95) to estimate residuals for the DGP1 setting  on the training set and on the test set for the pinball loss $(PL_\alpha)$ at different quantile levels $\alpha$ over 100 replications. The best performance for each quantile level on the test set is bold. The numbers in parentheses under a method’s name column are the corresponding empirical standard errors.}
\label{tab:alpha}
\begin{tabular}{llllll}
\hline
DGP & Measure & Case & $vineregRes$-0.05 & $vineregRes$-0.50 & $vineregRes$-0.95 \\\hline
\multirow{12}{*}{1} & \multirow{3}{*}{TPR} & 1 & 0.81 (0.02) & 0.80 (0.02) & 0.71 (0.02) \\
 &  & 2 & 0.81 (0.02) & 0.79 (0.03) & 0.76 (0.06) \\
 &  & 3 & 0.80 (0.02) & 0.78 (0.02) & 0.74 (0.03) \\  \hline
 & \multirow{3}{*}{FDR} & 1 & 0.09 (0.02) & 0.08 (0.01) & 0.06 (0.01) \\
 &  & 2 & 0.14 (0.02) & 0.13 (0.02) & 0.05 (0.01) \\
 &  & 3 & 0.15 (0.02) & 0.15 (0.02) & 0.09 (0.01) \\\hline
 & \multirow{3}{*}{Total Vars} & 1 & 4.63 (0.18) & 4.56 (0.20) & 4.00 (0.18) \\
 &  & 2 & 5.13 (0.25) & 5.04 (0.26) & 4.20 (0.17) \\
 &  & 3 & 4.99 (0.25) & 5.29 (0.33) & 4.40 (0.20) \\ \hline
 & \multirow{3}{*}{Time} & 1 & 0.13 (0.01) & 0.17 (0.01) & 0.11 (0.01) \\
 &  & 2 & 0.22 (0.01) & 0.30 (0.02) & 0.20 (0.01) \\
 &  & 3 & 0.54 (0.03) & 0.77 (0.05) & 0.44 (0.03)\\ \hline
\end{tabular}
\begin{tabular}{llllll}
\hline
DGP & Measure & Case & $vineregRes$-0.05 & $vineregRe$s-0.50 & $vineregRes$-0.95 \\
\hline
\multirow{9}{*}{1} & \multirow{3}{*}{$PL_{0.05}$} & 1 & \textbf{0.21 (0.01)} & \textbf{0.21 (0.01)} & \textbf{0.21 (0.01)} \\
 &  & 2 &\textbf{0.22 (0.01)}  & \textbf{0.22 (0.01)} & \textbf{0.22 (0.01)} \\
 &  & 3 & \textbf{0.21 (0.01)}  & \textbf{0.21 (0.01)} & \textbf{0.21 (0.01)} \\ \hline
 & \multirow{3}{*}{$PL_{0.50}$} & 1 & \textbf{0.79 (0.03)}& \textbf{0.79 (0.03)} & \textbf{0.79 (0.04)}  \\
 &  & 2 & 0.80 (0.03)  & \textbf{0.79 (0.03)} & 0.82 (0.04) \\
 &  & 3 & \textbf{0.76 (0.02)} & \textbf{0.76 (0.02)} & \textbf{0.76 (0.02)} \\\hline
 & \multirow{3}{*}{$PL_{0.95}$} & 1 &\textbf{0.43 (0.07)} &\textbf{0.43 (0.07)} & \textbf{0.43 (0.07)} \\
 &  & 2 & 0.41 (0.05) & \textbf{0.39 (0.04)} & 0.40 (0.05) \\
 &  & 3 &  \textbf{0.38 (0.03)}&  \textbf{0.38 (0.03)} & \textbf{0.38 (0.03)} \\ \hline
\end{tabular}
\end{table}
     
     \section{Pair copula families \& stopping \& complexity}\label{app4}
   \subsection{Candidate bivariate copula families}
   We work with following parametric bivariate copula families and their 90, 180, 270 degree rotations:	BB6, BB7, BB8, Clayton, Frank, Gaussian, Gumbel, Joe, t, and independence copula.
   \subsection{A stopping criterion}
   In the $s$th iteration of $vineregRes$ and $vineregParCor$, the conditional AIC of the model based on the D-vine copula $\hat{C}^{(s)}$ and D-vine copula density $\hat{c}^{(s)}$ with the model's effective degrees of freedom $|\bm{\hat{\Theta}}^{(s)}|$ is 
\begin{equation*}
CAIC(\bm{\hat{\Theta}}^{(s)}) = -2 \cdot \sum_{i=1}^{n} \left(\textrm{ln } \hat{c}^{(s)}_{0 | \bm{I}^{(s)}_{var}} (\hat{F}_Y(y_i) | \hat{F}_{p_1}(x_{i,p_1}), \ldots, \hat{F}_{p_{d_{(s)}}}(x_{i, p_{d_{(s)}}})) +  \textrm{ln } \hat{f}_Y(y_i) \right)+2 \cdot |\bm{\hat{\Theta}}^{(s)}|,
\end{equation*}
where it holds $\{p_1, \ldots, p_{d_{(s)}}\} \subseteq \bm{I}^{(s)}_{var}$.

Since we estimate the marginal distributions of the given data to have its pseudo-copula data non-parametrically, we consider the effective degrees of freedom as a penalization term, for which more details are provided in Section 7.6 of \textcite{hastie2009elements}. 
   \subsection{Complexity}
   In $vineregRes$, the selection of bivariate copulas exists at Steps 2 and 3. Step 2 fits a bivariate copula to the data of the iteration's pseudo-response and a candidate explanatory variable. Therefore, the number of bivariate copulas to be chosen at Step 2 is the total number of candidate explanatory variables at the given iteration. At the $s$th iteration, there are $p-(s-1)$ candidate explanatory variables. Thus, the number of bivariate copulas to be selected at this step is given by 
\begin{equation}
\sum_{s=1}^{p} (p-(s-1))= \frac{p \cdot (p+1)}{2}.
\label{eq:comp-step2}
\end{equation}
In $vineregParCor$, the variable selection is based on partial correlations; thus, its Step 2 does not perform any selection of bivariate copulas. 

Step 3 of $vineregRes$ and $vineregParCor$ extends the D-vine structure, adding the selected explanatory variable at Step 2. Thus, at the $s$th iteration, it results in the selection of $s$ pair copulas. Hence, Step 3 selects the following number of bivariate copulas:
\begin{equation}
\sum_{s=1}^{p} s= \frac{p \cdot (p+1)}{2}.
\label{eq:comp-step5}
\end{equation}
 Considering Equations \eqref{eq:comp-step2} and \eqref{eq:comp-step5}, we calculate the total complexity of $vineregRes$:
 \begin{equation}
\sum_{s=1}^{p} s + \sum_{s=1}^{p} (p-(s-1)) = p \cdot (p+1).
\label{eq:comp}
\end{equation}
Considering Equation \eqref{eq:comp-step5}, we calculate the total complexity of $vineregParCor$:
 \begin{equation}
\sum_{s=1}^{p} s= \frac{p \cdot (p+1)}{2}.
\label{eq:comp-parcor}
\end{equation}

Equations \eqref{eq:comp} and \eqref{eq:comp-parcor} show that the complexity of $vineregParCor$ and $vineregRes$ in terms of the total number of selected bivariate copulas is $\mathcal{O}(p^{2})$.

   \section{Irrelevant and redundant variables}\label{app5}
   $\bm{X}_\mathcal{I}$ is irrelevant for $Y$ given $\bm{X}_\mathcal{M}$ if $\bm{X}_\mathcal{I} \perp\!\!\!\perp(Y, \bm{X}_\mathcal{M})$. Since it holds that $\bm{X}_\mathcal{I} \perp\!\!\!\perp(Y, \bm{X}_\mathcal{M})$ if and only if $Y \perp\!\!\!\perp \bm{X}_\mathcal{I} | \bm{X}_\mathcal{M}$ and $\bm{X}_\mathcal{I} \perp\!\!\!\perp \bm{X}_\mathcal{M}$, it marginally implies $\bm{X}_\mathcal{I} \perp\!\!\!\perp Y$. Hence, $\bm{X}_\mathcal{I}$ do not have any predictability for $Y$ unconditionally or conditionally.
   
    $\bm{X}_\mathcal{R}$ is redundant for $Y$ given $\bm{X}_\mathcal{M}$ if $Y \perp\!\!\!\perp  \bm{X}_\mathcal{R} | \bm{X}_\mathcal{M}$ and $\bm{X}_\mathcal{M} \not\!\perp\!\!\!\perp \bm{X}_\mathcal{R}$. Conditional on $\bm{X}_\mathcal{M} = \bm{x}_\mathcal{M}$, for any $\bm{x}_\mathcal{M}$, $\bm{X}_\mathcal{R}$ does not provide any additional information for predicting $Y$, but $\bm{X}_\mathcal{R}$ without $\bm{X}_\mathcal{M}$ have some predictability for $Y$ 
   \subsection{Redundant variables in a D-vine copula}
\begin{proposition}{(Redundant variables in a D-vine copula)} If a $p$-dimensional D-vine copula is a $t$-truncated vine, where the response is a leaf node represented by the first node, $t >1$, $p \geq 3$, $t \leq p-1$, there are $p-t-1$ redundant or irrelevant variables.
\end{proposition}
\begin{proof}
We can write the conditional density of the response:
\begin{align*}
f_{1|2,\ldots,p}(x_1|x_2,\ldots, x_p)=  &\big[\prod_{j=3}^{p}  c_{1, j;2, \ldots j-1}\big( F_{1 | 2, \ldots j-1}(x_{1} | x_{2}, \ldots, x_{j-1}),   F_{j | 2, \ldots j-1}(x_{j} | x_2, \ldots, x_{j-1}\big)\big]\\
&\cdot c_{1,2}\big(F_1(x_1), F_2(x_2)\big) \cdot f_{1}(x_{1}).
\end{align*}
We can write the multiplication of copula density terms  as follows:
\begin{align*}
f_{1|2,\ldots,p}(x_1|x_2,\ldots, x_p)=  &\big[\prod_{j=3}^{t}  c_{1, j;2, \ldots j-1}\big( F_{1 | 2, \ldots j-1}(x_{1} | x_{2}, \ldots, x_{j-1}),  \\& F_{j | 2, \ldots j-1}(x_{j} | x_2, \ldots, x_{j-1}\big)\big]\\ 
&\cdot \big[\prod_{k=t+1}^{p}  c_{1, k;2, \ldots k-1}\big( F_{1 | 2, \ldots k-1}(x_{1} | x_{2}, \ldots, x_{k-1}), \\& F_{k | 2, \ldots k-1}(x_{k} | x_2, \ldots, x_{k-1}\big)\big]\\ 
&\cdot c_{1,2}\big(F_1(x_1), F_2(x_2)\big) \cdot f_{1}(x_{1}).
\end{align*}
A $t$-truncated vine has independence pair copulas in the tree levels  higher than $t$. Thus, the pair copula densities  in the tree levels  higher than $t$ are 1 everywhere in the domain. Accordingly, the multiplication of such pair copula densities are 1. Thus, we obtain
\begin{align*}
f_{1|2,\ldots,p}(x_1|x_2,\ldots, x_p)=  \big[&\prod_{j=3}^{t}  c_{1, j;2, \ldots j-1}\big( F_{1 | 2 \ldots j-1}(x_{1} | x_{2}, \ldots, x_{j-1}),  F_{j | 2, \ldots j-1}(x_{j} | x_2, \ldots, x_{j-1}\big)\big]\\ 
&\cdot c_{1,2}\big(F_1(x_1), F_2(x_2)\big) \cdot f_{1}(x_{1}).
\end{align*}
Without loss of generality, assume that it holds $p=k$. Then the conditional density of the response is written as
\begin{align*}
f_{1|2,\ldots,k}(x_1|x_2,\ldots, x_k)=  &\big[\prod_{j=3}^{k}  c_{1, j;2, \ldots j-1}\big( F_{1 | 2, \ldots j-1}(x_{1} | x_{2}, \ldots, x_{j-1}),   F_{j | 2, \ldots j-1}(x_{j} | x_2, \ldots, x_{j-1}\big)\big]\\
&\cdot c_{1,2}\big(F_1(x_1), F_2(x_2)\big) \cdot f_{1}(x_{1}).
\end{align*}
Take $k=t$, we have $f_{1|2,\ldots,p}(x_1|x_2,\ldots, x_p)= f_{1|2,\ldots,t}(x_1|x_2,\ldots, x_{t})$. If the variables $X_2, \ldots, X_{t}$ are dependent (independent) on the variables $X_{t+1}, \ldots, X_p$, there are $p-t-1$ redundant (irrelevant) variables for the response.
\end{proof}

\begin{example} Suppose that the first node and others represent the response and three explanatory variables in the D-vine in Figure 1 of the manuscript, respectively ($p=4$). Assume that the D-vine is a 2-truncated vine ($t=2$) and the copula for the pair (3, 4) is not independence. Then, there is one redundant variable for the response, which is the variable represented by the fourth node, given the two relevant variables represented by the second and third nodes.
\end{example}


\section{Simulation}\label{app6}
\subsection{Performance measures} 
\begin{equation*}
TPR(M) = \frac{\sum\limits_{\substack{X_j \in \mathcal{X}_{rel} \\ X_j \in \mathcal{X}^{M}_{chosen} \\ j=1, \ldots, |\mathcal{X}^{M}_{chosen} |}} 1 }{|\mathcal{X}_{rel} |}
\quad \textrm{and} \quad
FDR(M) = \frac{\sum\limits_{\substack{X_j \in \mathcal{X}_{irrel} \\ X_j \in \mathcal{X}^{M}_{chosen} \\ j=1, \ldots, |\mathcal{X}^{M}_{chosen} |}} 1 }{|\mathcal{X}^{M}_{chosen} |},
\end{equation*}
 where $\mathcal{X}_{rel}$ is the set of relevant variables,  $\mathcal{X}_{irrel}$ is the set of irrelevant variables, and $\mathcal{X}^{M}_{chosen}$ denotes the set of chosen variables by a method $M$. 
 
The accuracy of the quantile predictions $\hat{y}^{\alpha, M}_i$ at the level $\alpha$ by a method $M$ compared to the given response $y_i$,  $i=1,\ldots, n$ is given by the pinball loss as follows:
 \begin{equation*}
PL_{\alpha}(M)=\frac{\sum_{i=1}^{n} (\hat{y}^{\alpha, M}_i-y_i)(\mathbb{1}(y_i \leq \hat{y}_i^{\alpha, M})-\alpha) }{n}.
\end{equation*}

\subsection{Alternative methods}
 \textit{Penalized quantile regression with LASSO function (LQRLasso)}: it extends the linear quantile regression method, adding a LASSO penalty function in the quantile regression objective function to perform the variable selection. Thus, it solves the following optimization problem to estimate the coefficients at the quantile level $\alpha$:
\begin{equation}
\argmin_{\bm{\beta} \in \mathbb{R}^{p+1}} \bigg[\sum_{i=1}^{n} \rho_{\alpha}\left(y_{i}-\beta_0- \sum_{j=1}^{p} x_{i,j}\cdot \beta_j\right)+\sum_{j=0}^{p} \lambda\cdot|\beta_{j}|\bigg],
\label{eq:lasso}
\end{equation}
where $\rho_{\alpha}(r)=r(\alpha- \bm{1}(r<0))$ denotes the check function, $ \lambda \in \mathbb{R}$ corresponds to a tuning parameter for the penalization of the coefficients, $(y_i, x_{i,1}, \ldots, x_{i,p})$, $i=1, \ldots, n$, are realizations of the random vector $(Y, X_1, \ldots, X_p)$, and $Y$ denotes a response variable, and the others correspond to explanatory variables.

We perform a five-fold cross-validation using the training set to optimize $\lambda$ in each replication and set the associated coefficients below $10^{-8}$ to zero. The variables with zero coefficients are not relevant for predicting the response's quantiles. Such procedures are implemented in the  {\fontfamily{pcr}\selectfont R} package {\fontfamily{pcr}\selectfont rqPen} \parencite{sherwood2017rqpen}.  Since it is not straightforward to decide which transformation of variables or some interaction effects should be applied in a linear model in the high-dimensional data, we analyze the variables without any transformations. Since a LASSO penalty function is sensitive to the feature scales, we apply a standardization with zero mean and unit variance for each feature before its fit.

 \textit{Quantile regression forest (QRF)}: it is a nonlinear quantile regression method built upon random forest ideas \parencite{meinshausen2006quantile}. The core idea is to approximate the conditional distribution function of the response given explanatory variables.  It starts with fitting a random forest regression model to data for prediction. Then for each response observation to be predicted, it finds the terminal node, the node without any splits, in each tree. Next, each observation in the terminal node gets a weight inversely proportional to the total number of observations in the terminal node. Such weights are calculated for each tree for the associated observations and then normalized to one. Accordingly, each observation used for training the model gets a weight between zero and one, which sums up to one for all observations.  Finally, the weights are used to obtain the empirical conditional distribution function of the response given explanatory variables, where quantile regression estimates are calculated for the given observation to be predicted. The procedure is repeated for each new sample. More details are given by \textcite{meinshausen2006quantile}, and the ideas are implemented in the package {\fontfamily{pcr}\selectfont quantregForest} \parencite{Meinshausen2022}. We work with the package's default specifications; however, we set the minimum number of observations in terminal nodes to one-twentieth of the number of observations. We have not performed any cross-validation. 
 
$QRF$ assesses the feature importance implemented in the package {\fontfamily{pcr}\selectfont randomForest} \parencite{randomForest}. The idea is to calculate the decrease in the residual sum of squares (RSS) for each feature, which is the sum of the squares of the difference between the response values and the predicted response values. Specifically, the predicted values, thanks to the split through the feature, are considered in the RSS in a given tree and for a given feature. For instance, if the feature is split two times in a given tree, the decrease in the RSS is evaluated two times for that feature. The process is applied to all features and trees, and then the decrease in the RSS for each feature is divided by the number of trees grown to calculate the features' importance.

\begin{table}[H]
\caption{Estimated the LASSO penalty parameters $\lambda$ used in Section 4.3.}
\centering
\begin{tabular}{|l|l|l|l|}
\hline
Trait & Quantile & $G=100$ & $G=200$ \\ \hline
\multirow{3}{*}{FF} & 0.05 & 0.017 &  0.007\\ \cline{2-4} 
 & 0.50 & 0.014 & 0.010 \\ \cline{2-4} 
 & 0.95 & 0.024 &  0.009\\ \hline
\multirow{3}{*}{MF} & 0.05 & 0.029 & 0.007 \\ \cline{2-4} 
 & 0.50 & 0.009 & 0.011 \\ \cline{2-4} 
 & 0.95 & 0.032 & 0.029 \\ \hline
\multirow{3}{*}{PH\_V4} & 0.05 & 0.020 &0.018  \\ \cline{2-4} 
 & 0.50 & 0.027 & 0.014 \\ \cline{2-4} 
 & 0.95 &  0.013& 0.009 \\ \hline
\multirow{3}{*}{PH\_V6} & 0.05 & 0.015 & 0.010 \\ \cline{2-4} 
 & 0.50 & 0.011 & 0.009 \\ \cline{2-4} 
 & 0.95 & 0.013 & 0.013 \\ \hline
\end{tabular}
\end{table}

   \subsection{Simulated data sets}
   \begin{figure}[H]
     \caption{Pairs plots of a simulated data set of the response and relevant explanatory variables with 450 observations under the DGP1 (left) and DGP2 (right) settings, where diagonal: variable’s density estimates, lower: pairwise scatter plots with red curve showing a linear model fit and orange curve demonstrating a local polynomial regression fit.}
        \centering
        \begin{subfigure}[b]{0.45\textwidth}
          \centering
           \includegraphics[width=\textwidth]{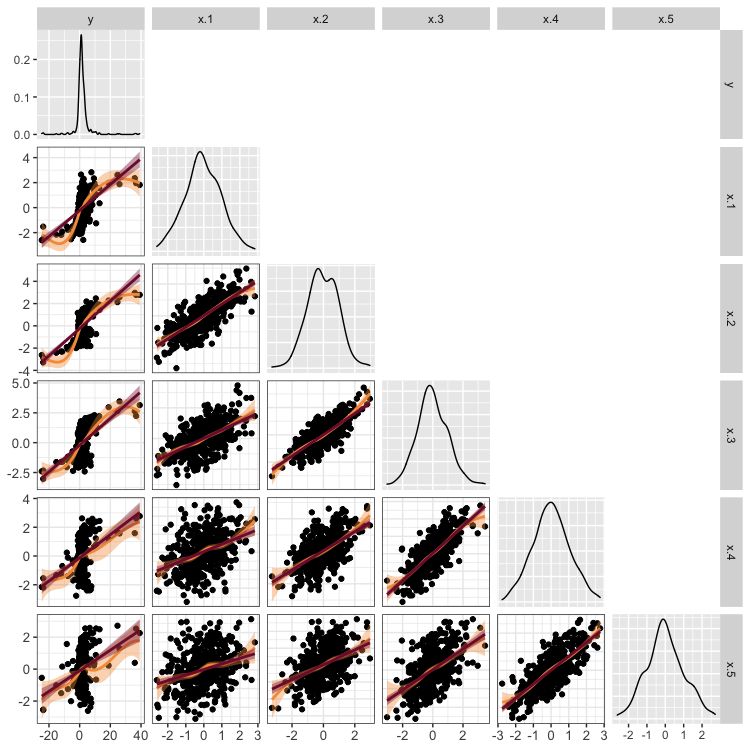}   
               \end{subfigure}
       \hspace{.5cm}
        \begin{subfigure}[b]{0.45\textwidth}  
          \centering 
           \includegraphics[width=\textwidth]{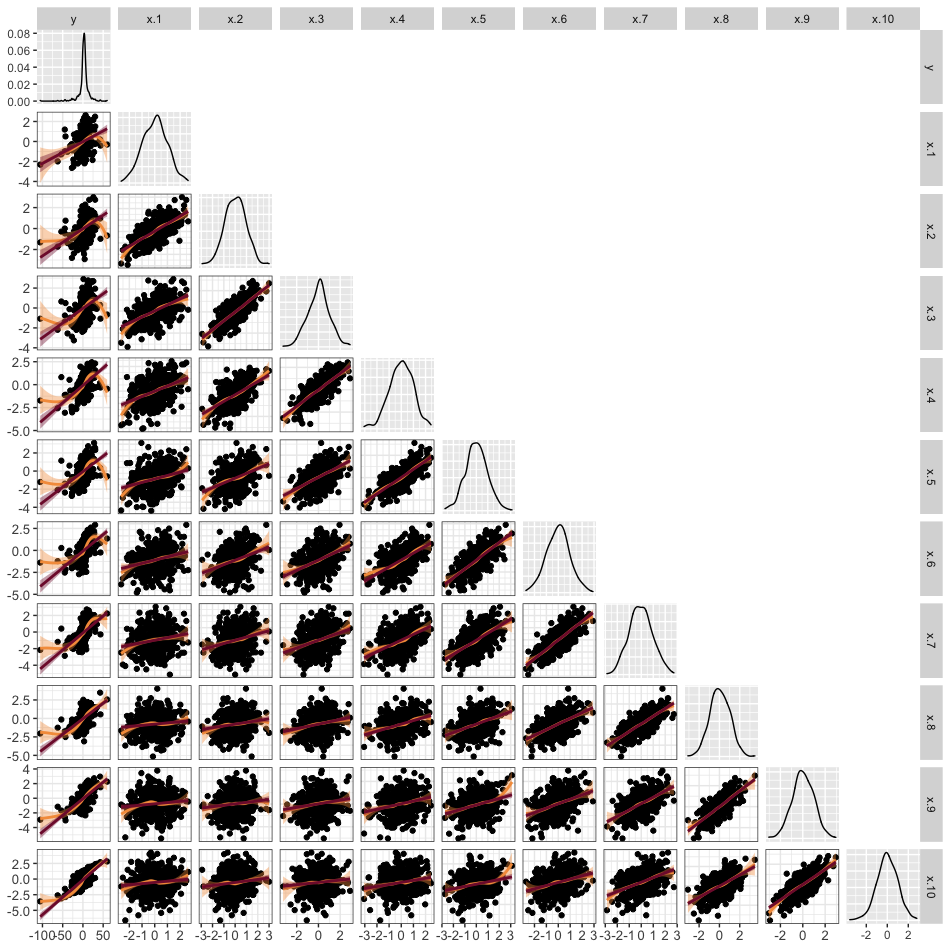}
        \end{subfigure}

  \label{fig:sim-x}
\end{figure}

\section{Simulation results for $QRF$}\label{app7}
We rank the variables based on their variable importance given by the model from highest to lowest for each case of both DGPs per replication. Then when we analyze the median rank of the selected variables out of 100 replications, we observe that $QRF$'s variable importance identifies the most important variables correct in both DGPs (available from the authors upon request). Also, \textcite{speiser2019comparison} study different variable selection methods for random forests and conclude that their performance varies regarding computational complexity and accuracy by analyzing different data sets. Thus, comparing selection methods for random forests, such as proposed by \textcite{genuer2010variable}, will be future work.

\section{Fitted D-vines' first tree level by $vineregRes$}\label{app8}
\begin{table}[H]
      \caption{Pair-copulas with the estimated Kendall's $\tau$ in the fitted D-vine's first tree for MF (left) and FF (right). The numbers under the edge denote the selected feature indices.}
      \begin{tabular}{llll}
\hline
$G$ & Edge & Family & $\tau$ \\ \hline
\multirow{12}{*}{100} &(MF, 7) &Gaussian	 &0.41\\
 & (7, 80) & Gaussian & 0.45 \\
 & (80, 82) & Gaussian & 0.45 \\
 & (82, 2) & BB7 & 0.33 \\
 & (2, 77) & BB7 & 0.30 \\
 & (77, 71) & Gumbel & 0.48 \\
 & (71, 66) & BB1 & 0.46 \\
 & (66, 70) & Gaussian & 0.46 \\
 & (70, 63) & BB1 & 0.46 \\
 & (63, 33) & BB8 & 0.42 \\
 & (33, 69) & Gaussian & 0.51 \\
 & (69, 57) & Gaussian & 0.51 \\ \hline
\multirow{8}{*}{200} & (MF, 4) & Gaussian & 0.41 \\
 & (4, 35) & Gaussian & 0.53 \\
 & (35, 1) & BB7 & 0.30 \\
 & (1, 29) & BB7 & 0.31 \\
 & (29, 46) & Gaussian & 0.59 \\
 & (46, 32) & BB7 & 0.49 \\
 & (32, 9) & Gumbel & 0.51 \\
 & (9, 40) & Gaussian & 0.54 \\ \hline
\end{tabular}
\hfill
\begin{tabular}{llll}
\hline
$G$ & Edge & Family & $\tau$ \\ \hline
\multirow{12}{*}{100} & (FF, 34) & Gaussian & 0.38 \\
 & (34, 164) & Frank & 0.51 \\
 & (164, 158) & Frank & 0.46 \\
 & (158, 137) & BB8 & 0.46 \\
 & (137, 18) & BB8 & 0.58 \\
 & (18, 20) & Gaussian & 0.73 \\
 & (20, 144) & Frank & 0.49 \\
 & (144, 23) & Frank & 0.50 \\
 & (23, 155) & BB8 & 0.55 \\
 & (155, 168) & Frank & 0.54 \\
 & (168, 51) & Frank & 0.48 \\ 
 & &&\\\hline
\multirow{8}{*}{200} & (FF, 19) & Gaussian & 0.38 \\
 & (19, 68) & BB8 & 0.58 \\
 & (68, 69) & BB8 & 0.60 \\
 & (69, 9) & Gaussian & 0.59 \\
 & &&\\
& &&\\
& &&\\
& &&\\ \hline
\end{tabular}
\end{table}

\begin{table}[H]
      \caption{Pair-copulas with the estimated Kendall's $\tau$ in the fitted D-vine's first tree for PH\_V4 (left) and PH\_V6 (right). The numbers under the edge denote the selected feature indices.}
\begin{tabular}{llll}
\hline
$G$ & Edge & Family & $\tau$ \\ \hline
\multirow{6}{*}{100} & (PH\_V4, 44) & Gaussian & 0.39 \\
 & (44, 150) & Frank & 0.57 \\
 & (150, 162) & BB8 & 0.47 \\
 & (162, 160) & BB8 & 0.45 \\
 & (160, 3) & BB8 & 0.57 \\
 & (3, 95) & Frank & 0.61 \\ \hline
\multirow{3}{*}{200} & (PH\_V4, 2) & Gaussian & 0.38 \\
 & (2, 69) & Frank & 0.66 \\
 & (69, 49) & Frank & 0.71 \\ \hline
\end{tabular}
\hfill
\begin{tabular}{llll}
\hline
$G$ & Edge & Family & $\tau$ \\ \hline
\multirow{6}{*}{100} & (PH\_V6, 1) & Gumbel & 0.34 \\
 & (1, 175) & BB8 & 0.40 \\
 & (175, 110) & BB8 & 0.53 \\
 & (110, 169) & BB8 & 0.58 \\ 
  & &&\\
& &&\\ \hline
\multirow{3}{*}{200} & (PH\_V6, 1) & BB7 & 0.34 \\
 & (1, 88) & Frank & 0.53 \\ 
& &&\\ \hline
\end{tabular}
\end{table}

\section{Quantile crossing of $LQRLasso$}\label{app9}
$LQRLasso$ performs variable selection and linear quantile regression at a specified quantile level at the same time. However, its final model fit might result in the crossing of quantile curves. For an illustration, we choose the first and fourth features obtained using $G=100$ as explanatory variables for predicting the trait MF at the levels 90\% and 95\% and let $LQRLasso$ run for both levels. Then only the fourth feature is identified in both, thereby estimating its coefficient. However, as seen in the following plot, the coefficient estimated by $LQRLasso$ leads to the quantile crossing.
   \begin{figure}[H]
     \caption{Fitted linear quantile regression curves at 90\% and 95\% levels for MF vs. a selected feature by $LQRLasso$, where candidate features to select are the first and fourth features using $G=100$.}
       \centering
                       \includegraphics[width=.75\textwidth]{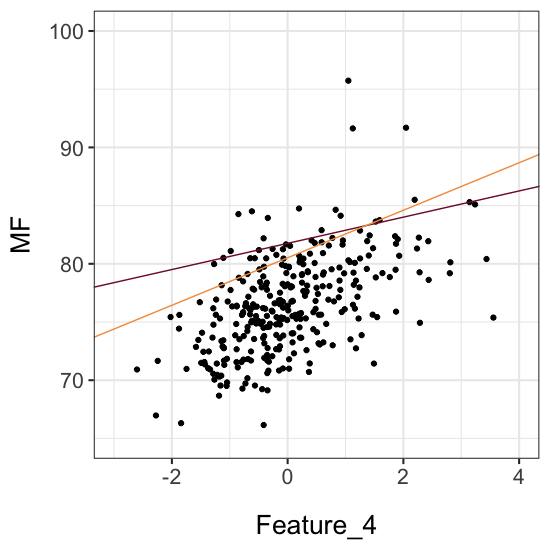}
   \label{fig:qcross-lqr}
\end{figure}

\section{Feature and method analyses of real data}\label{app10}
The minimum correlation between features is 0.58 (0.67) for FF, 0.11 (0.30) for MF, 0.45 (0.77) for  PH\_V4, and 0.55 (0.69) for PH\_V6, whereas its maximum is 0.98 (0.97) for FF, 0.91 (0.93) for MF, 0.97 (0.99) for PH\_V4, and 0.98 (0.99) for PH\_V6.

In the following, we provide selected feature indices by  $vineregRes$ and $vineregParCor$ as a result of our real data analysis in Section 4 of the manuscript. We remark that the features are based on the grouping $G$. For instance, the first feature using  a grouping size of $G=100\ (200)$ includes 100 (200) SNPs whose p-value in a simple linear regression for a trait is smaller than others.  Thus, selected features using  a grouping size of $G=100$ and $G=200$ cannot be compared. Moreover, $QRF$ uses all features to make predictions.

\begin{table}[ht]
\centering
\caption{The relevant feature indices by  $vineregRes$ and $vineregParCor$ for each grouping $G$ and trait. The same feature indices selected by  $vineregRes$ and $vineregParCor$ for each grouping $G$ and trait are highlighted.}
  \scalebox{0.75}{
\begin{tabular}{l|ll | ll}
\hline
 & $vineregRes$ & $vineregParCor$ & $vineregRes$ & $vineregParCor$ \\ \hline
Trait & \multicolumn{2}{c}{$G=100$} & \multicolumn{2}{c}{$G=200$} \\ \hline
\multirow{4}{*}{FF} & \multirow{4}{*}{\begin{tabular}[c]{@{}l@{}}34, 164, \textbf{158},\textbf{137}, 18, 20\\ 144, \textbf{23}, \textbf{155}, 168, \textbf{51}\end{tabular}} & \multirow{4}{*}{\begin{tabular}[c]{@{}l@{}}37, \textbf{158}, \textbf{137}, 68, 166, 163, \\ 140, \textbf{155}, 101, 104, 126, \textbf{23}, \\ 47, 80, 41, 42, 52, 55,\\  49, \textbf{51}, 136, 36\end{tabular}} & \multirow{4}{*}{\textbf{19}, 68, 69, \textbf{9}} & \multirow{4}{*}{\begin{tabular}[c]{@{}l@{}}\textbf{19}, 79, 24, \textbf{9}, 48, 84,\\ 70, 83, 82, 66, 63, 50, \\ 67, 45\end{tabular}} \\
 &  &  &  &  \\
  &  &  &  &  \\
 &  &  &  &  \\ \hline
\multirow{3}{*}{MF} & \multirow{3}{*}{\begin{tabular}[c]{@{}l@{}}\textbf{7, 80, 82, 2}, 77, 71, \\ 66, 70, \textbf{63}, 33, 69, \textbf{57}\end{tabular}} & \multirow{3}{*}{\begin{tabular}[c]{@{}l@{}}\textbf{7, 80, 2}, 45, 27, 24, \\ \textbf{57, 82}, 16, 52, 15, 76, \\ 62, \textbf{63}, 19, 81\end{tabular}} & \multirow{3}{*}{\begin{tabular}[c]{@{}l@{}}\textbf{4, 35, 1, 29}, 46, \\ 32, 9, \textbf{40}\end{tabular}} & \multirow{3}{*}{\begin{tabular}[c]{@{}l@{}}\textbf{4, 35, 1, 29, 40}, 3, \\ 38, 8, 14, 23, 20, 15, \\ 21\end{tabular}} \\
 &  &  &  &  \\
 &  &  &  &  \\ \hline
\multirow{3}{*}{PH\_V4} & \multirow{3}{*}{\textbf{44}, 150, 162, 160, 3, \textbf{95}} & \multirow{3}{*}{\begin{tabular}[c]{@{}l@{}}\textbf{44}, 5, 179, 153, 93, \textbf{95}\\ 92, 182, 51, 86, 88\end{tabular}} & \multirow{3}{*}{\textbf{2}, 69, 49} & \multirow{3}{*}{\begin{tabular}[c]{@{}l@{}}\textbf{2}, 77, 99, 26, 59, 73, \\ 4, 72, 20, 29, 83\end{tabular}} \\
 &  &  &  &  \\
 &  &  &  &  \\ \hline
\multirow{3}{*}{PH\_V6} & \multirow{3}{*}{\textbf{1}, 175, 110, 169} & \multirow{3}{*}{\begin{tabular}[c]{@{}l@{}}\textbf{1}, 26, 2, 119, 3, 116, \\ 86, 163, 72, 101, 130, 98\end{tabular}} & \multirow{3}{*}{\textbf{1, 88}} & \multirow{3}{*}{\begin{tabular}[c]{@{}l@{}}13, \textbf{1}, 82, 49, 66, 56, \\ 53, 30, 46, 9, 68, \textbf{88}\end{tabular}} \\
 &  &  &  &  \\
 &  &  &  &  \\ \hline
\end{tabular}}
\label{tab:genom-selvar}
\end{table}
Next, we provide selected feature indices by  $LQRLasso$ as a result of our real data analysis in Section 4 of the manuscript for each $G$, trait, and quantile level $\alpha$.  The intercept is included in all analyses by  $LQRLasso$. However, we observe that either there are one feature chosen by all methods or there are not any of them. In the analysis of the trait PH\_V6, the first feature is selected by all methods considered.
\begin{table}[H]
\centering
\caption{The selected feature indices by  $LQRLasso$ for each $G$, trait, and quantile level $\alpha$. The intercept is included in all.}
\begin{tabular}{l|l|l|l}
\hline
Trait & $\alpha$ & $G=100$ & $G=200$ \\
 \hline
\multirow{6}{*}{FF} & 0.05 & 76, 81, 83, 140, 157, 18, 163 & 9,28,40,41,42,69,79 \\  \cline{2-4}
 & 0.50 & \begin{tabular}[c]{@{}l@{}}3,4,23,37,46,49,52,54,61,63,68,\\ 71,78,80,85,88,99,100,101,104,105,107,108,\\ 126,130,132,134,138,142,144,147,151,155,\\ 158,160,163,164,166,168,171\end{tabular} & \begin{tabular}[c]{@{}l@{}}2,8,9,19,23,24,36,39,40,\\ 44,47,48,50,51,60,61,63,70,\\ 72,74,76,78,79,80,82,83,84,86\end{tabular} \\ \cline{2-4}
 & 0.95 & 17, 37, 88 & 9,19,49,61 \\
 \hline
\multirow{6}{*}{MF} & 0.05 & 2, 14, 24, 61, 63, 70 & 1,3,10,16,19,31,32,35,36,38,39 \\ \cline{2-4}
 & 0.50 & \begin{tabular}[c]{@{}l@{}}1,4,6,13,15,18,19,20,23,24,27,28,\\ 30,33,36,37,38,39,42,45,46,47,49,52,\\ 53,54,56,57,58,59,60,62,63,64,\\ 69,71,75,77,78,80,81,84,85,89\end{tabular} & \begin{tabular}[c]{@{}l@{}}1,2,3,4,7,15,23,25,29,30,\\ 35,38,40,41\end{tabular} \\ \cline{2-4}
 & 0.95 & 7, 8, 9, 63, 69, 70, 71 & 4,5,35,36 \\
 \hline
\multirow{4}{*}{PH\_V4} & 0.05 & 19, 95, 109, 125, 139, 171, 174, 189 & 35,44,48,79,80,94 \\ \cline{2-4}
 & 0.50 & \begin{tabular}[c]{@{}l@{}}2,5,44,90,130,139,141,146,\\ 153,158,179,181,196,197\end{tabular} & \begin{tabular}[c]{@{}l@{}}1,2,22,30,44,59,60,65,68,73,\\ 74,77,83,86,90,98\end{tabular} \\ \cline{2-4}
 & 0.95 & 51, 90, 95, 130, 139, 185, 191 & 48,65,86 \\
 \hline
\multirow{7}{*}{PH\_V6} & 0.05 & 1, 76, 93, 106, 127, 130, 175 & 1,38,43,53,55 \\ \cline{2-4}
 & 0.50 & \begin{tabular}[c]{@{}l@{}}1,2,10,15,22,25,29,32,39,43,45,53,\\ 59,65,70,72,73,77,86,88,91,95,96,\\ 97,98,99,101,106,110,116,123,125,\\ 128,130,131,132,138,142,156,159,\\ 160,163,165,169,177,178,179,180\end{tabular} & \begin{tabular}[c]{@{}l@{}}1,3,6,13,15,25,30,33,35,38,\\ 48,49,53,54,55,58,59,60,62,\\ 66,69,74,80,83,86,88,89,90\end{tabular} \\ \cline{2-4}
 & 0.95 & 1, 86, 135, 138 & 1,58,63,84,89 \\ \hline
\end{tabular}  
\label{tab:genom-selvar-LQR} 
\end{table}

\begin{table}[H]
\centering
\caption{The same selected feature indices by  $vineregRes$, $vineregParCor$, and $LQRLasso$ for each $G$ and trait and quantile level $\alpha$. (-) denotes the empty set.}
\begin{tabular}{l|l|l}
\hline
Trait & $G=100$ & $G=200$ \\ \hline
FF & - & 9 \\ \hline
MF & 63 & 35 \\ \hline
PH\_V4 & - & - \\ \hline
PH\_V6 & 1 & 1 \\ \hline
\end{tabular}
\end{table}

In the following table, $bicopreg$ denotes the vine copula regression on the first feature for PH\_V4 and PH\_V6. Thus, the D-vines contain two nodes: response and the first feature. Accordingly, we perform a bivariate a copula regression. Therefore, any variable selection by $vinereg$, $vineregRes$, and  $vineregParCor$ is not needed. The bivariate copula family selection between the response and the first feature is conducted as explained in Section 2.3 of the manuscript.

$bicopreg$ using  a grouping size of $G=200$ for PH\_V6 has a pinball loss of 0.97, 3.18, and 0.94 at $0.05, 0.50$, and $0.95$, respectively. Thus, $vineregRes$, which selects the first and 88th features (Table S4) as relevant for the prediction of PH\_V6, performs better than $bicopreg$. However, having more features in $vineregParCor$, including the first one, worsens the performance compared to $bicopreg$ using  a grouping size of $G=200$. On the other hand, $vineregParCor$ has better accuracy than $bicopreg$ using  a grouping size of $G=100$ for PH\_V6. Likewise, the mean of the selected feature indices using  a grouping size of $G=100$ for PH\_V4 is 102.33 for $vineregRes$ and 106.18 for $vineregParCor$. Further, both methods do not select the first feature as relevant but redundant for PH\_V4. Nevertheless, their prediction power is stronger than $bicopreg$, except at the level $0.50$ using $G=200$
\begin{table}[H]
\centering
  \caption{Comparison of the vine copula based regression methods' performance on the test set of the real data in Section 4 of the manuscript. The best performance on the test set for each quantile level, trait, and the grouping $G$ is highlighted.}
  \scalebox{0.75}{
\begin{tabular}{llrrr  rrr}
\hline
 &  & $vineregRes$ & $vineregParCor$ & $bicopreg$ & $vineregRes$ & $vineregParCor$ & $bicopreg$ \\ \hline
Trait & Measure & \multicolumn{3}{c}{$G=100$} & \multicolumn{3}{c}{$G=200$} \\ \hline
\multirow{4}{*}{PH\_V4} &  $PL_{0.05}$ & \hl{0.51} & \hl{0.51} & 0.53 & \hl{0.51} & \hl{0.51} & 0.52 \\
 & $PL_{0.50}$ &  1.93 & \hl{1.87} & 1.92 & 1.96 & 1.99 & \hl{1.93} \\
 &  $PL_{0.95}$&  \hl{0.56} & 0.58 & 0.58 & \hl{0.57} & \hl{0.57} & 0.58 \\\hline
\multirow{4}{*}{PH\_V6} & $PL_{0.05}$&\hl{1.01} & \hl{1.01} & \hl{1.01} & \hl{0.96} & 0.98 & 0.97  \\
 & $PL_{0.50}$ &\hl{3.09} & 3.10  & 3.18 & \hl{3.06} & 3.47 & 3.18\\
 & $PL_{0.95}$& 0.91 & \hl{0.89} & 0.93 & \hl{0.90} & 1.05 & 0.94\\ \hline
\end{tabular}}
\label{tab:f-index}
\end{table}

\raggedright

\printbibliography
\end{document}